\documentclass[a4paper,onecolumn,11pt,accepted=2024-11-09]{quantumarticle}
\pdfoutput=1
\usepackage[utf8]{inputenc}
\usepackage[english]{babel}
\usepackage[T1]{fontenc}

\usepackage{amsmath}
\usepackage{amsfonts}
\usepackage{amssymb}
\usepackage{amsthm}
\usepackage{mathtools}
\usepackage{mathrsfs}
\usepackage{stmaryrd}

\usepackage{color}
\usepackage{graphicx}
\usepackage{rotating}
\usepackage{caption}
\usepackage{subcaption}
\usepackage{ulem}
\usepackage{slashed}
\usepackage{braket}
\usepackage{tabu}
\usepackage{makecell}
\usepackage{diagbox}
\usepackage{placeins}

\usepackage{wasysym}

\usepackage[numbers]{natbib}

\usepackage{pgfplots}
\pgfplotsset{compat=1.17}

\usepackage[colorlinks,allcolors=black]{hyperref}

\usepackage{tikz}
\usetikzlibrary{calc}
\usetikzlibrary{positioning}
\usetikzlibrary{patterns}
\usetikzlibrary{matrix,backgrounds}
\pgfdeclarelayer{myback}
\pgfsetlayers{myback,background,main}
\tikzset{mycolor/.style = {line width=1bp,color=#1}}
\tikzset{myfillcolor/.style = {draw,fill=#1}}

\theoremstyle{definition}
\newtheorem{definition}{Definition}

\newtheorem{lemma}{Lemma}
\newtheorem{observation}{Observation}
\newtheorem{proposition}{Proposition}
\newtheorem{theorem}{Theorem}
\newtheorem{corollary}{Corollary}

\usepackage{cleveref}
\crefname{theorem}{Theorem}{Theorems}
\crefname{proposition}{Proposition}{Propositions}
\crefname{definition}{Definition}{Definitions}
\crefname{lemma}{Lemma}{Lemmas}
\crefname{figure}{Figure}{Figures}
\crefname{corollary}{Corollary}{Corollary}
\crefname{conjecture}{Conjecture}{Conjectures}
\crefname{section}{Section}{Sections}
\crefname{appendix}{Appendix}{Appendixes}
\crefname{observation}{Observation}{Observations}
\crefname{remark}{Remark}{Remarks}
\crefname{example}{Example}{Examples}
\crefname{equation}{Eq.}{Eqs.}
\crefname{table}{Table}{Tables}

\DeclareMathOperator{\tr}{tr}
\DeclareMathOperator{\id}{id}
\DeclareMathOperator{\Id}{Id}

\DeclareMathOperator{\AutG}{\mathbb{A}_G}
\DeclareMathOperator{\AutH}{\mathbb{A}_H}
\DeclareMathOperator{\AutGin}{\mathbb{A}_G^{ \text{in}}}
\DeclareMathOperator{\AutGout}{\mathbb{A}_G^{ \text{out}}}
\DeclareMathOperator{\AutGoutupp}{\overline{\mathbb{A}_G^{ \text{out}}}}

\newcommand{\defeq}{\mathrel{\mathop:}=}

\newcommand{\axilleaf}{\includegraphics[width=3.5mm]{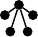}}
\newcommand{\contwins}{\includegraphics[width=3.5mm]{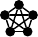}}
\newcommand{\discontwins}{\includegraphics[width=3.5mm]{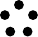}}
\newcommand{\single}{\includegraphics[width=3.5mm]{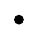}}

\renewcommand{\figurename}{\textbf{Figure}}

\newcommand{\ie}{\textit{i.e.}}
\newcommand{\cf}{\textit{cf.}}
\newcommand{\eg}{\textit{e.g.}}

\definecolor{violet}{HTML}{53257F} 
\definecolor{green}{HTML}{257a7f}
\definecolor{green2}{HTML}{527f27}
\definecolor{brown}{HTML}{852e29}

\colorlet{nodecolor}{black}
\colorlet{fadeoutnodecolor}{nodecolor!50}
\colorlet{edgecolor}{black}
\colorlet{fadeoutedgecolor}{edgecolor!50}
\colorlet{edgecolorU}{edgecolor}
\colorlet{edgecolorV}{fadeoutedgecolor}
\colorlet{edgecolorW}{fadeoutedgecolor}
\colorlet{circlecolor}{brown}

\newcommand{\new}[1]{{\color{black} #1}}

\newcommand\Odd[1]{\text{Odd}\left( #1  \right) }

\newlength{\dhatheight}
\newcommand{\hathat}[1]{%
    \settoheight{\dhatheight}{\ensuremath{\hat{#1}}}%
    \addtolength{\dhatheight}{-0.3ex}%
    \hphantom{\hat{#1}} % Ensures baseline alignment and spacing
    \kern -0.2em % Slightly shifts everything back to the left
    \hat{\vphantom{\rule{1pt}{\dhatheight}}}%
    \smash{\kern -0.48em \hat{#1}}}

\begin{document}

\title{The Foliage Partition: An Easy-to-Compute LC-Invariant for Graph States}

\author{Adam Burchardt}
\affiliation{QuSoft, CWI and University of Amsterdam, Science Park 123, 1098 XG Amsterdam, the Netherlands}
\orcid{0000-0003-0418-257X}
\email{adam.burchardt.uam@gmail.com,}

\author{Frederik Hahn}
\affiliation{Electrical Engineering and Computer Science Department, Technische Universit{\"a}t Berlin, 10587 Berlin, Germany}
\affiliation{Dahlem Center for Complex Quantum Systems, Freie Universit{\"a}t Berlin, 14195 Berlin, Germany}
\orcid{0000-0002-9349-4075}

\email{mail@frederikhahn.eu}
% \homepage{frederikhahn.eu}

\thanks{\newline\newline Both authors contributed equally to this work.}

\maketitle

\begin{abstract}
This paper introduces the foliage partition, an easy-to-compute LC-invariant for graph states, of computational complexity $\mathcal{O}(n^3)$ in the number of qubits. 
Inspired by the foliage of a graph, our invariant has a natural graphical representation in terms of leaves, axils, and twins. 
It captures both, the connection structure of a graph and the $2$-body marginal properties of the associated graph state. We relate the foliage partition to the size of LC-orbits and use it to bound the number of LC-automorphisms of graphs.
We also show the invariance of the foliage partition when generalized to weighted graphs and qudit graph states.  
\end{abstract}

\section{Introduction}

Graph states play an essential role in quantum information theory and quantum computing~\cite{gottesman1997stabilizer,PhysRevLett.91.107903,PhysRevA.68.022312, PRXQuantum.1.020325}.
As examples of stabilizer states, graph states can be directly implemented on quantum computers~\cite{graphstatesIBMQ, 9866745, briegel2009measurement}. 
As their name suggests, graph states can be visually represented by graphs, with vertices corresponding to qubits and edges representing controlled two-qubit operations for the experimental implementation of the graph state. 
In the language of graph theory, local Clifford (LC) operations acting on a graph state correspond to local complementations of the associated graph.

Three main questions surrounding graph states include their entanglement properties~\cite{PhysRevA.69.062311,Entanglement_GraphStates,HajdusekMurao_conference}, their equivalence under local operations~\cite{BOUCHET199375,AlgorithmLULC,
vandennestGraphicalDescriptionAction2004,Clifford,
ketkarNonbinaryStabilizerCodes2006,BurchardtRaissi20,RaissiBurchardt22}, and their classification~\cite{Clifford,Danielsen_2008,Adcock2020mappinggraphstate}.
Considerable effort has been made to address these questions:
first, an efficient algorithm to determine the LC-equivalence of given graph states with complexity $\mathcal{O}(n^4)$ in the number of qubits $n$ has been discovered \cite{BOUCHET199375,AlgorithmLULC,Bouchet1991AnEA}; 
second, a finite set of invariants which completely characterizes the LC-equivalence class of any graph state --and more generally, any stabilizer state--  has been identified \cite{PhysRevA.71.022310,invStab}; 
third, graph states have been classified up to 12 qubits~\cite{Danielsen_2008} and there are lower and upper bounds on the number of non-equivalent graph states for a given number of qubits \cite{Clifford}. 
Our focus is on the second problem as we introduce an easy-to-compute LC-invariant for graph states.

Although known invariants completely characterize the LC-equivalence class of any graph state, they are computationally inefficient \cite{PhysRevA.71.022310,invStab}, requiring knowledge of the given state's full stabilizer set, which size is exponential in the number of its qubits. 
This inefficiency renders these known invariants impractical. 
To mitigate this inefficiency, we introduce the \textit{foliage partition}, a partition of the graph's vertices induced by a simple equivalence relation.
With computational complexity $\mathcal{O}(n^3)$ in the number of qubits $n$, the foliage partition is an easy-to-compute LC-invariant that eliminates the need to compute the exponential stabilizer set. 
Closely related to the foliage of a graph~\cite{Hahn_2022, Dahlberg_2020, Hahn_2019}, the foliage partition has a simple graphical representation in terms of leaves, axils, and twins.
We show that foliage partition is invariant under local complementation of the associated graph, and hence that it is an LC-invariant for graph states.
We also explore the relationship between the foliage partition and both the entanglement properties and symmetries of the associated graph state. 
In particular, we show, that the foliage partition of the graph is trivial if and only if the corresponding state is 2-uniform, \ie~all $2$-body marginals are maximally mixed. Furthermore, we provide constraints on the LC-automorphisms group of a graph in terms of its foliage partition.

Remarkably, the foliage partition invariance extends to the generalization from qubits to \textit{qudits}: we prove its invariance under the generalized qudit local complementation operations for \textit{weighted graphs}.
Therefore, the foliage partition of any (weighted) graph constitutes an easy-to-compute LC-invariant of the corresponding (qudit) graph state.

The paper is organized as follows:
\cref{Section:qubits} recalls the concept of graph states and presents prior results concerning their LC-equivalence.
\cref{invariantQubit} introduces the \textit{foliage partition} of a graph. It provides two independent definitions of the foliage partition and highlights its basic properties. 
\cref{FoliageRepresentation} defines \textit{foliage graph} and \textit{foliage representation}, a representation of graphs that captures the structure of the foliage partition.
\cref{Sec:ent} discusses the relationship between the foliage partition of a graph and the entanglement properties of the associated graph state. 
\new{\cref{Section:asymptotic} provides the lower bound on the number of LC-non-equivalent  unlabelled graphs, also called LC-classes.} 
\cref{Section:symmetry_of_graphs} considers the LC-automorphism group of a graph in the context of the introduced invariant.
\cref{algo} presents an efficient algorithm to find the foliage partition of a given graph, which runs in $\mathcal{O}(n^3)$ time in the number of graph vertices $n$. 
\cref{Section:qudits} extends the notion of foliage partition to the qudit graph state setting.
In \cref{proofs} proves that the foliage partition is well-defined and LC-invariant for both qubit and qudit graph states.
As a supplement to this paper, we provide complete tables of foliage partitions, foliage graphs, and LC-automorphism groups for all $146$ LC-orbits of graphs up to $8$ qubits. A detailed explanation of supplementary material can be find in \cref{Section:foliage_of_small_graphs}.

\section{LC-equivalence of graph states}
\label{Section:qubits}

Recall that, a graph $G=(V,E)$ consists of a finite set of vertices $V$ and of edges $E\subseteq V \times V$. A graph is called \textit{simple} if it contains neither edges connecting a vertex to itself nor multiple edges between the same vertices. A simple graph is uniquely defined by its \textit{adjacency matrix} $\Gamma_G =(a_{ij})_{i,j \in V}$, where $a_{ij}=1$ if vertices $i,j$ are connected, and $a_{ij}=0$ otherwise.

The \textit{graph state} $\ket{G}$ of a simple graph $G$ is defined as
\begin{equation}
\label{graph_state}
\ket{G} \defeq \prod_{(v,w)\in E} \textbf{CZ}^{\{vw\}} \ket{+}^{\otimes V},
\end{equation}
where $\textbf{CZ}^{\{vw\}} \defeq  (-1)^{ij} \ket{ij}_{vw} \bra{ij}$ is a controlled-Z operator on qubits $v$ and $w$, and $\ket{+}\defeq \frac{1}{\sqrt{2}} (\ket{0}+\ket{1})$.

Recall further, that the Pauli group is generated by local Pauli operators, and the Clifford group, its normalizer, is the largest group preserving the Pauli group. 
The Clifford group is generated by Hadamard, CNOT, and phase gates, while the local subgroup, the \textit{local Clifford group}, is generated by Hadamard and phase gates alone \cite{nielsen_chuang_2010}. 
As the local Clifford group does not contain entangling operations, it retains quantum information properties. 
Two graph states are said to be \textit{LC-equivalent} if they can be transformed into each other using only local Clifford operators, and it is known that they are LC-equivalent if their associated graphs are equivalent under \textit{local complementation operations} \cite{vandennestGraphicalDescriptionAction2004}.

 \begin{definition}[Local complementation]\label{def:local_complementation} 
 For every vertex $a \in V$ of a graph $G=(V, E)$, we define a \textit{locally complemented graph} $\tau_{a}(G)$ with adjacency matrix
	\begin{equation}\label{eq:local_complementation_adjacency}
	\Gamma_{\tau_{a}(G)} \defeq \Gamma_{G}+\Gamma_{K_{N_{a}}} \quad (\bmod 2),
	\end{equation}
	where $K_{N_{a}}\defeq (V, {\{(v,w)~\mid~\left(v,w\in N_{a}\right) \wedge \left(v \neq w\right)\}})$ 
	is the complete graph on the neighborhood $N_{a}$ of $a$ and empty on all other vertices.
\end{definition}

\begin{theorem}[see~\cite{vandennestGraphicalDescriptionAction2004}]
\label{LCisLC}
Two graph states are LC-equivalent if and only if their associated graphs are equivalent under a sequence of local complementations.
\end{theorem}

 It is useful to use the acronym LC-equivalent both for the local Clifford equivalence and for the local complementation equivalence of the related graphs due to \cref{LCisLC}.

As an example, consider two graphs on $n$ vertices: the star graph $S_n$, which has one central vertex and $n-1$ leaves (vertices with just one neighbor), and the fully connected graph $K_n$ (clique). 
These graphs are related by local complementation (see \cref{figStarClicque}), meaning their associated graph states $\ket{K_n}$ and $\ket{S_n}$ are also LC-equivalent.
In fact, both states are also locally equivalent to $\ket{\text{GHZ}_n}$, a maximally entangled state of $n$ qubits \cite{Graph1}.

\begin{figure}[ht]
\centering
\includegraphics{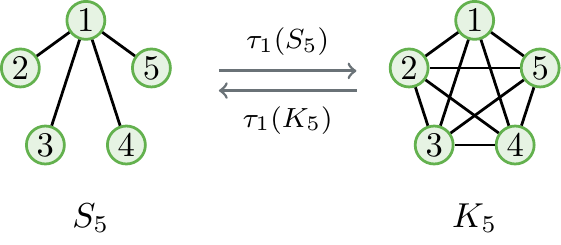}
\caption{The fully connected graph $K_n$ and the star graph $S_n$ are equivalent under the local complementation operation with respect to the center of the star.}
\label{figStarClicque}
\end{figure}

\section{Foliage partition}
\label{invariantQubit}

% For any graph, we can create a binary relation on its vertices, which leads to a partition of the vertex set.
For any graph, there is a simple binary relation on its vertices that leads to a partition of the vertex set.

\begin{definition}
\label{def:relation}
Consider a graph and its adjacency matrix $\Gamma_G =(a_{ij})_{i,j \in V}$. 
Two vertices are related $v \sim w$ if and only if they are in the same connected component and 
$
a_{v,u_1}\cdot a_{w,u_2}=a_{v,u_2}\cdot a_{w,u_1}
$
for any other pair of vertices $u_1,u_2$.
This relation is an equivalence relation and thus defines a partition on the set of vertices by grouping related vertices into subsets $V_i$\footnote{Recall that, a \textit{partition} of a set $V$ is a collection of its subsets $V_i \subseteq V$, such that they are pairwise disjoint, that is, $V_i \bigcap V_j =\emptyset$ for $i\neq j$, and such that they sum up to the whole set as $\bigcup_{i=1}^k V_j =V$. Any equivalence relation $\sim$ on $V$ provides its partition into equivalence classes defined as $[v]\defeq \{w\in V: w\sim v\}$ for $v\in V$.}. We call this partition the \textit{foliage partition} of the graph, and denote it by $\hat{V} \defeq \{V_1,\ldots,V_k\}$.
We say that the foliage partition is \textit{trivial} if $\hat{V} = \{\{v\}\mid v \in V\}$.
\end{definition}

\cref{figPartition} presents the foliage partition of several graphs on five vertices. A remarkable property of the foliage partition is its invariance under LC-operations.

\begin{theorem}
\label{th1q}
The foliage partition of $G$ is LC-invariant.
\end{theorem}

\noindent
We refer to \cref{proofs}, where we present an exhaustive proof of \cref{th1q}, and show that the relation $\sim$ is, indeed, an equivalence relation. As an immediate consequence of \cref{th1q}, we have the following corollary.

\begin{corollary}
\label{necessary_condition}
Let $G_1$ and $G_2$ be graphs on the same vertex set.
Having the same foliage partition is a necessary condition for graphs $G_1$ and $G_2$ to be LC-equivalent.
\end{corollary}

\begin{figure}[ht]
\centering
\includegraphics{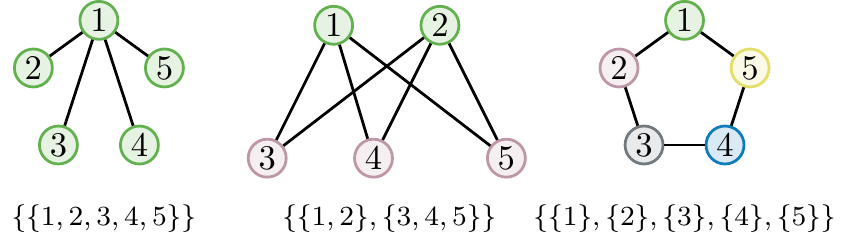}
\caption{The fully connected graph $K_5$, the bipartite graph $K_{2,3}$, and the circular graph $C_5$ are shown with their foliage partition. Since their foliage partitions are different, the graphs cannot be LC-equivalent.}
\label{figPartition}
\end{figure}

We now will examine the possible structures of subsets within a foliage partition and will discover that there are essentially three possible structures for any foliage partition set that contains at least two elements.

\begin{proposition}
\label{def:foliage_partition}
Each set $V_i\in \hat{V}$ in the foliage partition is of 
exactly\footnote{In the unique case of a fully connected graph with two vertices $K_2$, the foliage partition contains one set. This set has a single element that is both of type \axilleaf~and \contwins. However, this exceptional case should not create confusion later on.} 
one of the following types:
\begin{enumerate}
\item[\single] $V_i$ contains only one element, \ie~$|V_i|=1$, or if not:
\item[\axilleaf] $V_i$ is a star graph and is connected to the rest of the graph only by its center,
\item[\contwins] $V_i$ is a fully connected graph and all its vertices have pairwise the same neighborhood,
\item[\discontwins] $V_i$ is a fully disconnected graph and all its vertices have the same  neighborhood.
\end{enumerate}
Each such set $V_i$ is maximal, meaning that adding any vertex not in $V_i$ to create a larger set $V_i'$ would result in a set that is not one of the types listed above.
Conversely, every subset $W\subset V$ that is of one of the types above and that is maximal, belongs to the foliage partition, \ie~$W\in \hat{V}$.
\end{proposition}

\new{We refer to \cref{proofs}, where we present an exhaustive proof of \cref{def:foliage_partition}. }
\cref{def:foliage_partition} allows us to define the \textit{type function} 
$
T:\hat{V} \rightarrow \{\single, \axilleaf, \contwins, \discontwins\}
$ 
associating each set $V_i\in \hat{V}$ in the foliage partition with its type. 
To better understand the types \single, \axilleaf, \contwins, and \discontwins, refer to the examples provided in \cref{fig:types_of_foliage_partition_subsets}. 

Local complementation can alter the types of subsets, as shown in  \cref{fig:types_of_foliage_partition_transformation}
and formalized in \cref{LCinLLC}. 
However, the foliage partition is invariant under local complementation. Informally speaking, the property of sharing the same neighborhood or being a star subset turns out to be LC-invariant. 

\begin{figure*}
\centering
		\includegraphics[width=0.23\textwidth]{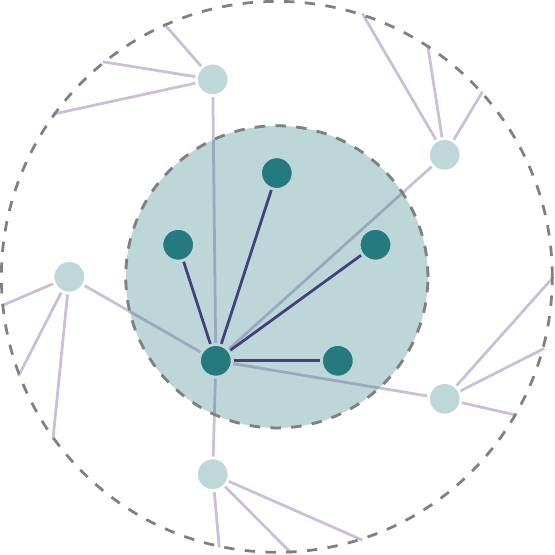}
 \hfill
		\includegraphics[width=0.23\textwidth]{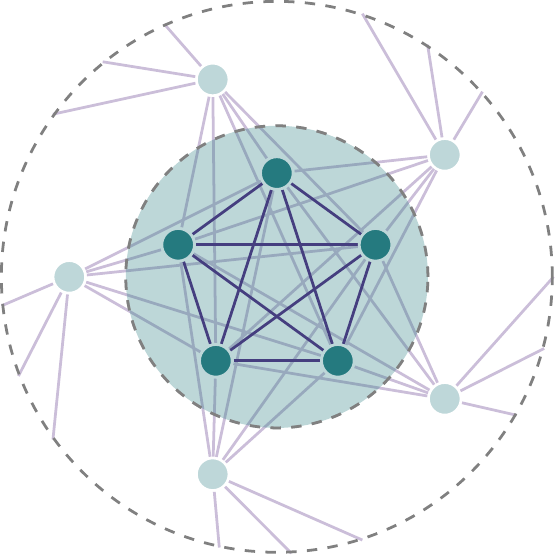}
 \hfill
		\includegraphics[width=0.23\textwidth]{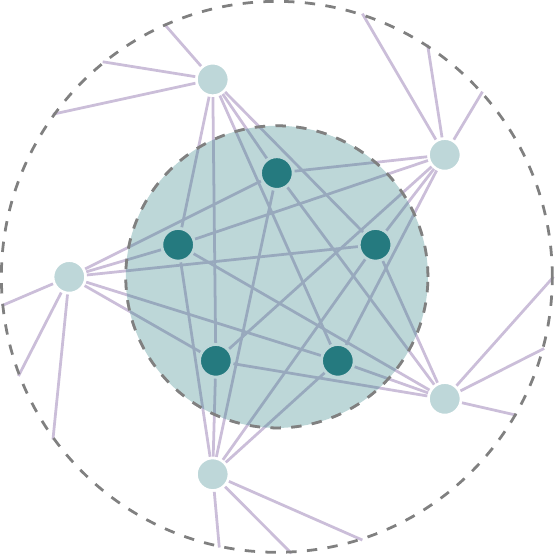}
 \hfill
		\includegraphics[width=0.23\textwidth]{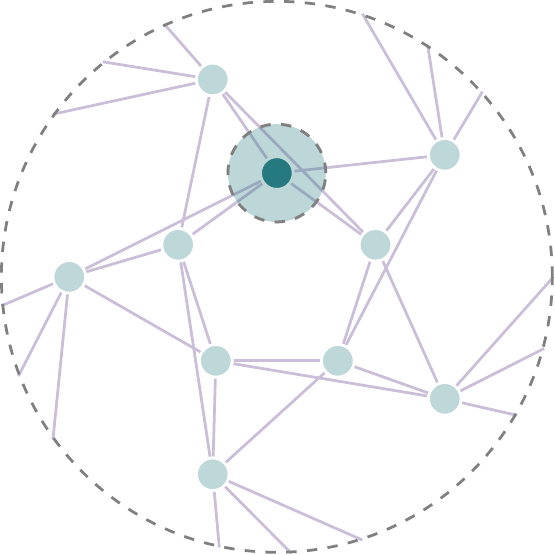}
\caption{The possible types for sets $V_i\in \hat{V}$ in the foliage partition of graphs are shown in the green-shaded inner circles.  From the left: a star graph connected to the rest of the graph only by its center, a fully connected graph with vertices sharing pairwise the same neighborhood, a fully disconnected graph with vertices sharing the same neighborhood, and finally a one-element set.}
\label{fig:types_of_foliage_partition_subsets}
\end{figure*}

\begin{figure*}
\centering
\includegraphics[width = \textwidth, 
    trim=0cm % left
        9cm % bottom
        0cm % right
        7.75cm,% top
    clip]{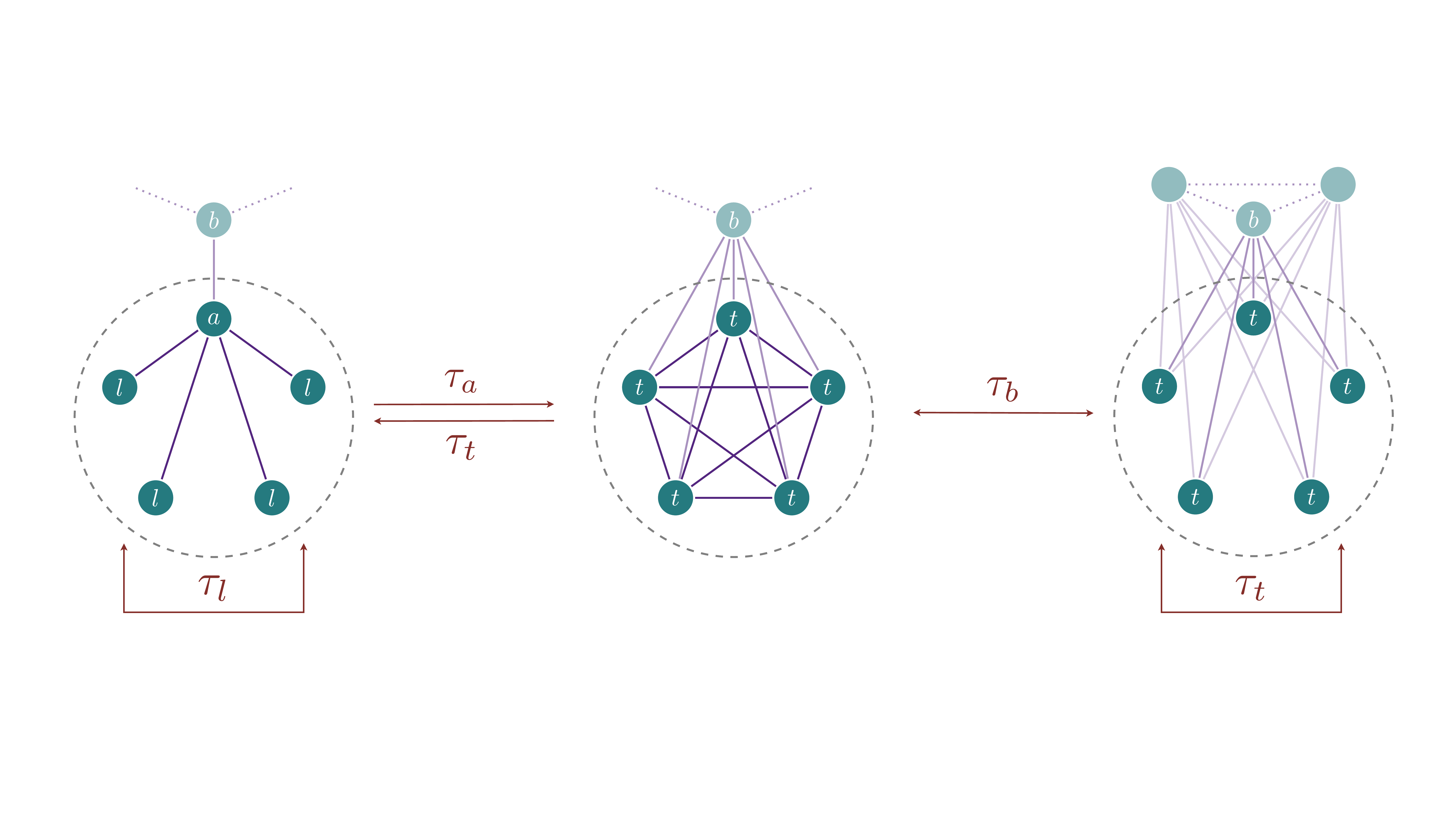}
\caption{
Transformation of the types of sets in a foliage partition. A local complementation $\tau_{a}$ with respect to the axil $a$ of \axilleaf~transforms the type \axilleaf~into the \mbox{type \contwins.} Conversely, any vertex $t$ of \contwins~can be transformed with $\tau_{t}$ into the axil of \axilleaf. Likewise, a local complementation $\tau_{b}$ with respect to a neighbor $b$ of a set of disconnected twins \discontwins~transforms the type \discontwins~into the type \contwins. Conversely, any vertex $t$ of \contwins~can be transformed with $\tau_{t}$ into a disconnected twin in \discontwins. Local complementations with respect to a leaf in type \axilleaf~or a twin in type \discontwins~leave the sets invariant.
For simplicity not all neighbors of $b$ are shown.
}
\label{fig:types_of_foliage_partition_transformation}
\end{figure*}

\begin{figure*}
\centering
		\includegraphics[width=\textwidth, 
    trim=1cm % left
        7cm % bottom
        1cm % right
        7cm,% top
    clip]{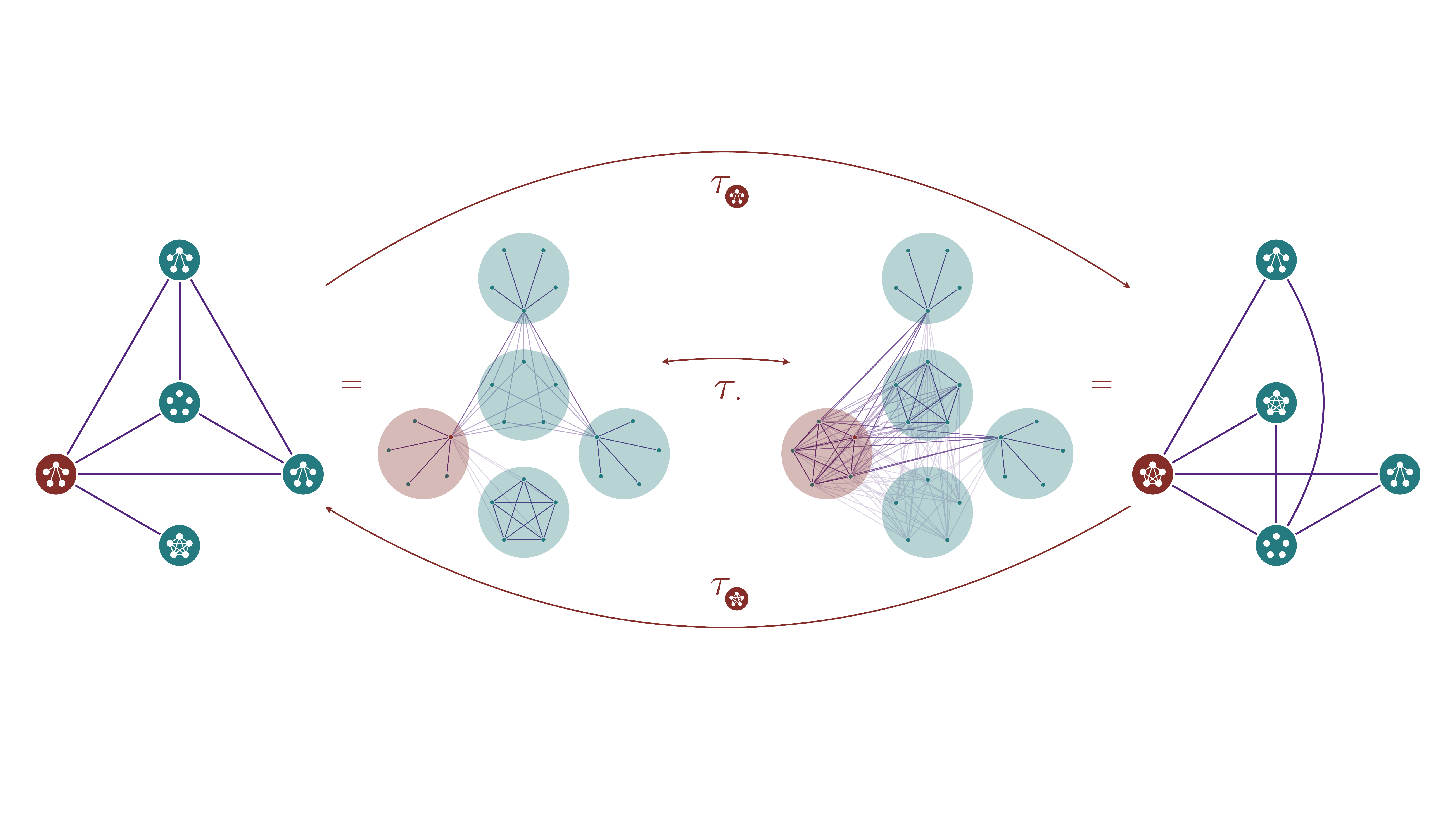}
\caption{Lifted local complementation operations on the foliage representation a graph.
	The two middle graphs are related via 
	local complementation with respect to the node 
	highlighted in dark brown, while the two outer 
	graphs are the foliage representations of the 
	two inner ones and are related via 
	lifted local complementation
	with respect to the node (set)
	highlighted in light brown.
 The concept of lifting local complementations to the foliage representation is formalized in \cref{LCinLLC}.
}
\label{fig:lifted_foliage_LC}
\end{figure*}
As we reach the end of this section, it is pertinent to touch upon the name foliage partition. 
The \textit{foliage} of a graph represents the comprehensive collection of all leaves, axils, and twins. Here, a \textit{leaf} denotes a vertex with only one neighbor, and an \textit{axil} is defined as the unique neighbor of a leaf. \textit{Twins} are vertices that share the same neighborhood, demonstrating a common interconnectedness within the graph.

It is noteworthy to mention that the foliage being an LC-invariant was previously established \cite{Hahn_2022}. However, the focus of the current exploration lies in identifying the nuanced relationship between the foliage and the \textit{foliage partition}. This relationship is captured by two direct consequences of \cref{def:foliage_partition}, forming the foundation for the following observations.

% We conclude this section by commenting on the name foliage partition. The \textit{foliage} of a graph is the collection of all all leaves (a \textit{leaf} is a vertex with just one neighbor), axils (an \textit{axil} is the unique neighbor of a leaf) and twins (two vertices are \textit{twins} if they share the same neighborhood). 
% Note that the foliage of a graph was already known to be LC-invariant \cite{Hahn_2022}. 
% A relation between foliage and foliage partition can be captured in the following direct consequences of \cref{def:foliage_partition}.
\begin{observation}
\label{extraOb1}
Two vertices $v_1,v_2 \in V $ belong to the same set in the foliage partition $\hat{V}$ if and only if either they are twins, or a leaf and an adjacent axil, or a pair of leaves which are adjacent to the same axil. 
\end{observation}
\begin{observation}
\label{extraOb2}
The union of all sets in the foliage partition that are at least of size two is the foliage of $G$.
\end{observation}

As we continue our exploration, we unravel the intricate details of the foliage partition and how it can be used to represent graphs in a more compact way.

\section{Foliage representation}
\label{FoliageRepresentation}

The foliage partition provides a convenient way to represent any graph $G$ by grouping together all vertices that belong to the same set in the foliage partition. First, we introduce the concept of a \textit{foliage graph}. This graph encapsulates information about the connections between different sets in the foliage partition of the original graph.

\begin{definition}[Foliage graph]
\label{def.foliage_graph}
Consider a graph $G=(V,E)$ with the foliage partition $\hat{V}=\{V_1,\ldots,V_k \}$. We define its foliage graph $\hat{G}=(\hat{V}, \hat{E} )$ on the set of vertices given by the foliage partition $\hat{V}$. Two vertices $V_i,V_j\in \hat{V}$ are connected if and only if the subsets $V_i,V_j$ are connected in the initial graph $G$. In other words, $v_i\in V_i$ and $v_j\in V_j$ exist such that $(v_i, v_j)\in E$. 
\end{definition}

The foliage graph, along with the type function, captures most information about the underlying graph $G$. Note, however, that for $V_i\in \hat{V}$ of the type $T(V_i)=\axilleaf$, the information which vertex $v\in V_i$ is the axil is lost in both $\hat{G}$ and $T$.  
To address this, we introduce the \textit{axil set} of $G$, denoted as $A\subset V$, which contains all axils of $G$. The axil set $A$ includes exactly one element from each $V_i$ where $T(V_i)=\axilleaf$, and no elements from other subsets.

\begin{definition}[Foliage representation]
\label{def.foliage_representation}
To a given graph $G$ we associate the tuple $ (\hat{G},T,A)$ consisting of its foliage graph $\hat{G}$, its type function $T$, and its axil set $A$, and denote it as the \textit{foliage representation} of the graph $G$.
\end{definition}

The foliage representation of a graph is uniquely defined and faithfully represents the corresponding graph. 

\begin{proposition}
The foliage representation of a graph allows to recover the original graph.
\end{proposition}

\begin{proof}
Let $(\hat{G},T,A)$ be a tuple, where $\hat{V}$ is a partition of a set $V$, $\hat{G}$ is a graph on $\hat{V}$, $T: \hat{V}\rightarrow \{\single, \axilleaf, \contwins, \discontwins\}$ is an arbitrary function, and $A \subset V$ is a subset containing exactly one element from each $V_i$ such that $T(V_i)=\axilleaf$, and no elements from other subsets. 

We shall define the graph $G$ on the set of $V$ vertices by defining the set of its edges. Consider any two vertices $v,w \in V$. If both vertices $v,w$ belong to the same element $V_i$ in $\hat{V}$, we shall connect them by an edge if and only if either $T(V_i)=\contwins$, or both $T(V_i)=\axilleaf$ with either $v\in A$ or $w \in A$. Suppose that vertices $v,w$ belong to different elements in $\hat{V}$, \ie~$v\in V_i, w\in V_j$, $i\neq j$. We shall connect $v$ and $w$ by an edge if and only if the following three conditions are simultaneously satisfied: 
(i) $V_i, V_j$ are connected in $\hat{G}$, 
(ii) $T(V_i )= \axilleaf \Rightarrow v \in A$, 
(iii) $T(V_j )= \axilleaf \Rightarrow w \in A$.

\new{
We will now verify that the map defined in the previous paragraph recovers the initial graph. Consider an arbitrary graph $G$ on the set of vertices $V$, its foliage partition $\hat{V}$, and its foliage representation $(\hat{G}, T, A)$. Moreover, consider two vertices $v, w \in V$. Suppose first that the vertices $v, w \in V$ belong to the same set in the foliage partition $v, w \in V_i$. 
%or to two different sets $v \in V_i, w \in V_j$, with $i \neq j$. 
In this case, $v$ and $w$ are connected in $G$ if and only if either $T(V_i) = \contwins$, or both $T(V_i) = \axilleaf$ with either $v \in A$ or $w \in A$. Note that the map defined in the previous paragraph captures exactly this condition. A similar analysis can be performed in the second case, where $v$ and $w$ belong to different sets in the foliage partition.
}
\end{proof}

Since local complementation operations do not change the foliage partition, the foliage graphs of any starting graph and of a locally complemented graph derived from this starting graph are defined on the same set of vertices. Below, we show how local complementation operations $\tau_a (G)$ change the foliage representation of a graph $G$. An example of this lifted local complementation is shown in~\cref{fig:lifted_foliage_LC}.

\begin{observation}[Lifted local complementation]
\label{LCinLLC}
Consider the graph $G=(V,E)$ with corresponding foliage representation $ (\hat{G}, T, A)$. Let vertex $a\in V$, and $ V_a\in \hat{V}$ be the element in $ \hat{V} $ such that $a\in V_a$. The foliage representation of the locally complemented graph $\tau_a (G)$ is given by $ (\tau_{V_a} (\hat{G}), T', A')$, where $\tau_{V_a} ( \hat{G})$ is the locally complemented graph $\hat{G}$ with respect to its vertex $V_a$, and $ T', A'$ are modified  type function and axil set given by
\begin{equation}
T' (V_j)\defeq 
\begin{cases} 
\axilleaf & \text{if } j=a \text{ and } T(V_a)= \contwins,\\ 
\contwins & \text{if } j=a \text{ and } T(V_a)= \axilleaf,\\ 
\discontwins &\text{if }V_a,V_j\in \hat{E}\text{ and }T(V_j)= \contwins,\\ 
\contwins &\text{if }V_a,V_j\in \hat{E}\text{ and }T(V_j)= \discontwins,\\ 
T(V_j) & \text{otherwise,} 
\end{cases} 
\label{Tprime}
\end{equation} 
for all $V_j \in \hat{V}$ and
\begin{equation} 
A' \defeq 
\begin{cases} 
A \cup \{a\} & \text{if } T(V_a)= \contwins,\\ 
A \setminus \{a\} & \text{if } T(V_a)= \axilleaf,\\ 
A & \text{otherwise.} 
\end{cases} 
\label{Aprime}
\end{equation} 
\end{observation}

As a direct consequence of \cref{LCinLLC}, we obtain the following necessary condition for LC-equivalence, which is stronger than the one presented in \cref{necessary_condition}.

\begin{corollary}
Consider two graphs $G_1$ and $G_2$ on the same set of vertices with the same foliage partition. 
% The graphs $G_1$ and $G_2$ can only be LC-equivalent if the corresponding foliage graphs $\hat{G_1}$ and $\hat{G_2}$ are also LC-equivalent.
If the graphs $G_1$ and $G_2$ are LC-equivalent the corresponding foliage graphs $\hat{G_1}$ and $\hat{G_2}$ are also LC-equivalent.
\end{corollary}

Note that the foliage graph $\hat{G} $ generally has a smaller number of vertices than the initial graph $G$, so it is easier to explore its LC-equivalence class. The foliage graph of a graph with trivial foliage partition is isomorphic to the graph itself.
Having understood the basic properties of foliage graphs, we can delve deeper into their structure by introducing the idea of recursion in this context.

Foliage graphs can be defined recursively. 
Given the foliage graph $\hat{G}$ of $G$, we can define the \textit{$2$nd-foliage graph} $\hathat{G}$ as the foliage graph of $\hat{G}$. 
Obviously, we can continue and may define  the \textit{$k$th-foliage graph} for any integer $k$. 
However, this continuation saturates at a certain point. At this point, for some integer $k$, all $k'$th-foliage graphs for $k'\geq k$ become isomorphic to this $k$th-foliage graph. Consequently, this process can be seen as a step-by-step reduction of the given graph to a graph with a smaller number of vertices until the reduction saturates.

We can now ask at which step this reduction saturates (saturation time) and how many vertices the final graph has (saturation size). 
For example, the graph presented in \cref{fig:graph139} has saturation time $k=2$, since its $2$nd-foliage graph is the first with a trivial foliage partition. This graph has saturation size $6$, since its $2$nd-foliage graph has six vertices. The saturation time of the graphs presented in \cref{figPartition} is $k=1,2,0$ for $K_5,K_{2,3},C_5$ respectively. 

Now that we have looked at specific examples of saturation times, it is also interesting to examine how these times vary across different families of graphs.

Different families of graphs have different saturation times. For instance, line graphs, which are made up of vertices connected in a straight line (graphs with edges $(i,i+1)$ for $i=1,\ldots,n-1$), have a saturation time $k= \lfloor \tfrac{n}{2}\rfloor$ that increases linearly as the size $n$ of the graph increases. However, for fully connected graphs, where every vertex is connected to every other vertex, the saturation time $k=1$ is always the same, no matter how big the graph is.

% In general, the saturation time scales differently for different families of graphs. On the one hand, for example, for line graphs (graphs on $n$ vertices with edges $(i,i+1)$ for $i=1,\ldots,n-1$), the saturation time scales linearly $k= \lfloor \tfrac{n}{2}\rfloor$ with the size $n$. On the other hand, for fully connected graphs the saturation time is constant $k=1$ regardless of the size of the graph $n$.

\begin{figure}[ht]
\centering
	\begin{subfigure}[b]{0.15\textwidth}
		\centering
		\includegraphics[width=\textwidth]{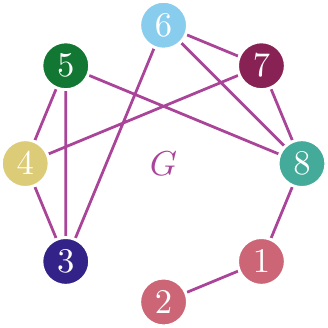}
		%\caption{$G$}%\label{subfig:Graph139} 
	\end{subfigure}
 	\begin{subfigure}[b]{0.15\textwidth}
		\centering
		\includegraphics[width=\textwidth]{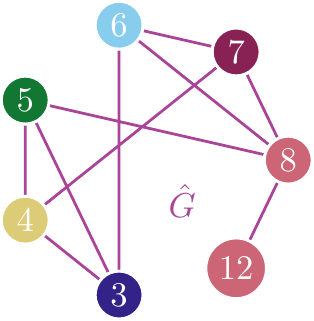}
		%\caption{$\hat{G}$}%\label{subfig:139foliagegraph} 
	\end{subfigure}
 	\begin{subfigure}[b]{0.15\textwidth}
		\centering
		\includegraphics[width=\textwidth]{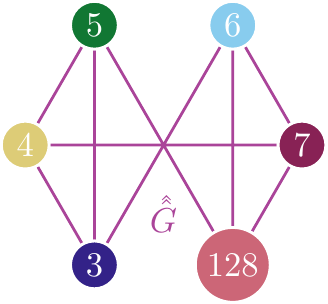}
	   %\caption{$\hathat{G}$}%\label{subfig:139Secondfoliagegraph} 
	\end{subfigure}
 \caption{\label{fig:graph139} Graph $G$, its foliage graph $\hat{G}$, and $2$nd-foliage graph $\hathat{G}$ corresponding to the representative of LC-orbit (139); \cf~\cref{fig:graph_representatives}. For simplicity, we omit the brackets in the description of vertices in the two foliage graphs, \eg~the label $128$ represents the set $\{1,2,8\}$.}
\end{figure}

Let us continue to consider the properties of saturation time and saturation number. It is important to note that both are LC-invariant, making them orbital properties rather than characteristics of individual graphs alone.
We can extend saturation time and number further when considering the averages across all LC-orbits.

By taking the average of both the saturation time and saturation number over all LC-orbits of graphs with a fixed number of vertices $n$, we can define what we will call the \textit{$n$-th saturation time} and the \textit{$n$-th saturation size}.
In addition to these definitions, we can also look at the proportion of certain types of orbits.

We introduce the \textit{ratio of reducible orbits} of size $n$. This is defined as the number of LC-orbits with non-trivial foliage partition divided by the total number of LC-orbits of graphs with $n$ vertices. Similarly, we define the \textit{ratio of fully reducible orbits} of size $n$, calculated as the number of LC-orbits for which any $k$-th foliage partition is a one-vertex graph divided by the total number of LC-orbits of graphs with $n$ vertices.

The averaged saturation quantities and orbit ratios provide valuable insight when applied to the analysis of LC-orbits of graph states.

We present the above quantities for all LC-orbits of graph states with $n\leq 8$ in \cref{tab:stabilization}. These measures carry useful information about the extent to which LC-orbits of graphs with $n$ vertices can be reconstructed from LC-orbits of graphs with fewer vertices.
We believe that studying the asymptotic behavior of these measures can shed some light on the general structure of LC-orbits of large graphs.

\begin{table}[ht]
\begin{center}
% \begin{tabular}{|c|| c |c| c|c |c| c|} 
%  \hline
%   $n$ & 3 & 4 & 5 & 6 & 7 & 8\\ [0.5ex] 
%  \hline\hline
%  $n$-th saturation time & 1 & 1.5 & 1.25 & 1.55 & 1.77 &1.60  \\ 
%  \hline
%  $n$-th saturation size & 1 & 1 & 2 & 2.27 & 3.12 & 4.26 \\
%  \hline
%  ratio (reducible orbits) & 1 & 1 & 0.75 & 0.82 & 0.85 & 0.84 \\
%  \hline
%  ratio (fully reducible orbits) & 1 & 1 & 0.75 & 0.73 & 0.58 & 0.42 \\
%  \hline
% \end{tabular}
		\includegraphics[width=0.55\columnwidth, 
    trim=0cm % left
        0cm % bottom
        0cm % right
        0cm,% top
    clip]{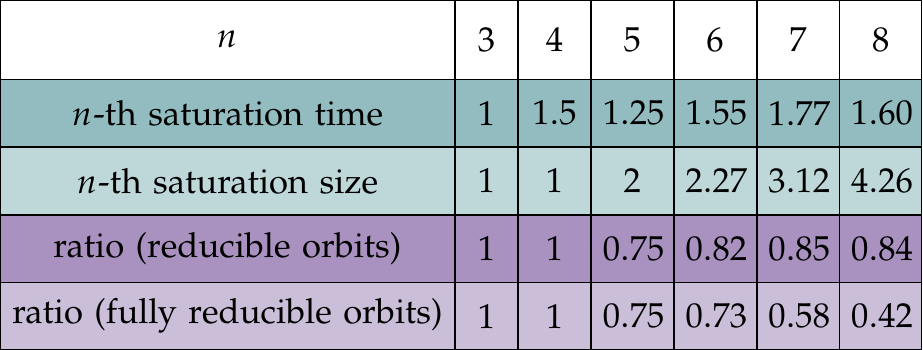}
    \caption{The average value of \textcolor{green}{saturation time and saturation size} of LC-orbits, and the \textcolor{violet}{ratio of (fully) reducible orbits}, rounded to the second decimal. For example, the LC-orbits of graphs with $6$ vertices can be reduced to graphs with $2.27$ vertices in $1.55$ steps on average. The probability that an LC-orbit has non-trivial foliage partition is $0.82$, and the probability that it reduces to single-element graph is $0.73$.}
    \label{tab:stabilization}
    \end{center}
\end{table}

We conclude this section by introducing the notion of a \textit{normal form} of a graph $G$. 

Consider any graph $G$, its foliage partition $\hat{V}$ and its foliage representation $(\hat{G},T,A)$. By locally complementing with respect to any axil $a\in A$, we change the structure of $\hat{G},T$, and $A$ according to \cref{LCinLLC}. In particular, this reduces both the number of type $\axilleaf$~sets in $\hat{V}$ and the size of the axil set $A$ by one; see \cref{Tprime,Aprime}. Therefore, local complementation with respect to all elements in the axil set $A$ (in any order) produces an LC-equivalent graph $G'$ with the same foliage partition $\hat{V}$ and different foliage representation $(\hat{G'},T',A')$, where the function $T$ takes only values in the set $\{\discontwins,\contwins,\single\}$,~while $A'=\emptyset$.\footnote{Note that such $G'$ is not always unique. It depends on the order of local complementations applied on the axil set $A$ of $G$}. Therefore, we have the following observation.

\begin{observation}
\label{NorFor}
Every graph $G$ is LC-equivalent to a graph $G'$, for which the type function $T'$ takes values only in the set $\{\discontwins,\contwins,\single\}$ and for which the axil set $A'=\emptyset$ is empty. We call such $G'$ a \textit{normal form} of $G$.
\end{observation}

Normal forms of this type will be crucial for analyzing the entanglement structure of graph states (see \cref{prop2} in the following section).

\section{Entanglement structure}
\label{Sec:ent}

The examination and understanding of the entanglement properties of graph states have garnered considerable attention and sparked numerous studies in the recent past \cite{Graph1,Entanglement_GraphStates,gittsovichMultiparticleCovarianceMatrices2010,RaissiBurchardt22}. It has been established that the entanglement properties of a given graph state $\ket{G}$ can be effectively derived from its adjacency matrix $\Gamma_G$. 

Here, we demonstrate that the entanglement properties inherent to the graph state $\ket{G}$ can be comprehensively encapsulated by the adjacency matrix $\Gamma_{\hat{G}}$ of the corresponding foliage graph $\hat{G}$, in conjunction with the type function $T$. 

We initiate our discussion by recapitulating the well-established findings pertaining to the entanglement properties of graph states.

To quantify the entanglement in graph states we use the \textit{entropy of entanglement} for bipartitions of the graph state's quantum system into two parts \cite{PhysRevA.64.022306,Entanglement_GraphStates}. The entropy of entanglement of a state $\ket{\psi}\in \mathcal{H}_2^{\otimes n}$ with respect to the bipartition $A|A'$ -- where $A'$ is the complement of the set $A$ -- is determined by the reduced density matrix $\rho_A\defeq\text{tr}_{A'} \ket{\psi}\bra{\psi}$ via the expression
\begin{equation}
S_A (\ket{\psi})\defeq- \text{Tr}  \big(\rho_A \;\text{log}_2 \rho_A  \;\big).
%\text{log}_2\; \Big(\text{rank}\; \big(\rho_A \big)\Big).
\end{equation}
For graph states, the entropy of entanglement is an integer-valued measure, 
aligning with continuous entanglement measures such as the Schmidt rank or the purity of the reduced density matrices~\cite{Entanglement_GraphStates}. 
Importantly, the value of $S_A  (\ket{\psi})$ is bounded between $0$ and $\text{min} \{|A|,|A'|\}$, 
and the above definition is symmetric with respect to the exchange of $A$ and $A'$ 
\cite{PhysRevA.64.022306,PhysRevA.69.062311}. 

It is known that for graph states $\ket{G}$, the entropy of entanglement can be computed efficiently from the adjacency matrix of the corresponding graph $G$. 
Specifically, it has been proven that
\begin{equation}
\label{rankkk}
S_A (\ket{G})= \text{rank}\; \Gamma^{A\,A'}_G,
\end{equation}
where $\Gamma^{A\,A'}_G$ denotes the submatrix of the adjacency matrix $\Gamma_G$ with 
rows and columns indexed by subsets $A$ and $A'$ respectively, and the rank is computed over the field $\mathbb{Z}_2$ \cite[Proposition 3]{PhysRevA.69.062311}. 

Based on this, we can define the \textit{entropy vector}, also known as the \textit{Schmidt vector} \cite{entropycones1,entropycones2}:
\begin{equation}
\label{Schmidt_vector}
S (\ket{\psi})  =\big( S_A (\ket{\psi}) \, :\, A\subset [n]\big),
\end{equation}
given by an ordered collection of the entropies of entanglement for all possible subsets $A\subset [n]$. The length of the vector is given by $|S (\ket{\psi}) |=2^n$ and it characterizes the entanglement properties for all bipartitions of a given state. 

%\subsection{Entanglement in graph states revisited}

As we shall see, the adjacency matrix $\Gamma_{\hat{G}}$ of the foliage graph $\hat{G}$ together with the type function $T$ contains all information about the entanglement properties of the underlying graph state $\ket{G}$. We will show it explicitly for the graphs of the normal form, see \cref{NorFor}\footnote{One may establish such a connection for an arbitrary graph $G$ by translating the sequence of LC-operations on its axil set $A$ into changes in the foliage graph's adjacency matrix $\Gamma_{\hat{G}}$ and the type function $T$.}. 

First, we reduce the information given by the type function $T$ to a diagonal matrix 
\begin{align}
D_{T}&\defeq\text{diag}\big( d_{V_i}
\,\big| \, V_i\in \hat{V} \big),\\
d_{V_i}&\defeq
\begin{cases} 
1 & \text{if } T(V_i)=\contwins \\ 
0 & \text{if } T(V_i)= \single,\discontwins
\end{cases}
\end{align}
indexed by the elements of the foliage partition $V_i\in \hat{V}$. 

Second, we define the following square matrix 
\begin{equation}
  E_{T,\hat{G}}\defeq \Gamma_{\hat{G}}+D_{T}  
\end{equation}
with rows and columns indexed by the elements of the foliage partition $V_i,V_j\in \hat{V}$
and find the following.

\begin{proposition}
\label{prop2}
Consider any graph $G=(V,E)$ in a normal form (see \cref{NorFor}) with its corresponding foliage graph $\hat{G}$ and its type function $T$. 
For an arbitrary subset $A\subset [n]$, we define the following two subsets of the foliage partition $\hat{V}$ set\footnote{Note that $\hat{A'}$ and $\hat{A}$ are subsets of $\hat{V}$, which are no longer complementary. Indeed, there may exist $V_i\in \hat{V}$ with non-trivial intersections with both sets $A$ and $A'$.}:
\begin{align*}
\hat{A} \defeq&\{V_i\in \hat{V} : V_i \cap A\neq \emptyset \},\\
\hat{A'} \defeq&\{V_i\in \hat{V} : V_i \cap A' \neq \emptyset \}.
\end{align*}
Then the entanglement entropy of $\ket{G}$ can be expressed as
\begin{equation}
S_A (\ket{G}) = \text{rank}\, E_{T,\hat{G}}^{\hat{A}\,\hat{A'}} ,
\end{equation}
where $E_{T,\hat{G}}^{\hat{A}\,\hat{A'}}$ is a submatrix of $E_{T,\hat{G}}$ containing rows indexed by elements from $\hat{A}$ and columns indexed by elements form $\hat{A'} $. 
\end{proposition}

\begin{proof}
Consider any bipartition $A|A'$ of the vertex set $V$. Using \cref{rankkk}, we shall show that $\text{rank}\, \Gamma^{A\,A'}=E_{T,\hat{G}}^{\hat{A}\,\hat{A'}}$. 
The proof consist of two parts. 

Firstly, we show that $\text{rank}\, \Gamma^{A\,A'}=\text{rank}\, \Gamma^{B\,B'}$, where $B\subset A,B'\subset A'$ are specific subsets containing at most one element from each element $V_i\in \hat{V}$ in the foliage partition. Note that if the set $V_i\in \hat{V}$ is of either of the types $T(V_i)=\contwins,\single,\discontwins$ the rows in $\Gamma^{A\,A'}$ corresponding to $v\in V_i\cup A$ are all the same. Therefore, one can remove all but one without changing the rank of the matrix. The same holds true for the columns in $ \Gamma^{A\,A'}$ corresponding to $v\in V_i\cup A'$. By assumption, $G$ is in a normal form, hence there are no elements $V_i\in \hat{V}$ of the type $T(V_i)=\axilleaf$.

Secondly, we show that the matrices $\Gamma^{B\,B'}$ and $E_{T,\hat{G}}^{\hat{A}\,\hat{A'}}$ are isomorphic, and hence that they have the same rank. In the previous paragraph, we have shown that $\text{rank}\, \Gamma^{A\,A'}=\text{rank}\, \Gamma^{B\,B'}$, where $B\subset A$ contains exactly one element $v_i\in A\cup V_i$ for any $V_i\in \hat{V}$ such that $A\cup V_i\neq \emptyset$. Similarly, $B'\subset A'$ contains exactly one element $v_i\in A'\cup V_i$ for any $V_i\in \hat{V}$ such that $A'\cup V_i\neq \emptyset$. This allows us to establish a bijection between the sets $B$ and $\hat{A}$, and $B'$ and $\hat{A'}$, simply by identifying the elements $v\in V_i$  with the sets $V_i$ in the foliage partition to which they belong. 

We shall see that this construction leads to an isomorphism between matrices $\Gamma^{B\,B'}$ and $E_{T,\hat{G}}^{\hat{A}\,\hat{A'}}$. Indeed, consider any element in $\Gamma^{B\,B'}$ indexed by a pair of vertices $v,w\in V$, and the respective element in $E_{T,\hat{G}}^{\hat{A}\,\hat{A'}}$ indexed by $V_i,V_j \in \hat{V}$, where $v\in V_i, w\in V_j$. 

Consider the following two cases: $V_i=V_j$, or $V_i\neq V_j$. 

In the first case, it is easy to see that in both matrices $\Gamma^{B\,B'}$ and $E_{T,\hat{G}}^{\hat{A}\,\hat{A'}}$ the respective element equals $1$ when $T(V_i)=\contwins$ and equals $0$ when $T(V_i)=\discontwins$. It is not possible that $T(V_i)=\single$, as in this case $|V_i|=1$, and $V_i$ cannot have a non-trivial intersection simultaneously with $B$ and $B'$. 

In the second case, when $V_i\neq V_j$, one can see that the respective elements in both matrices are equal to $(\Gamma_{\hat{G}})_{V_i,V_j}$. This shows that matrices $\Gamma^{B\,B'}$ and $E_{T,\hat{G}}^{\hat{A}\,\hat{A'}}$ are isomorphic, hence have the same rank, which concludes the proof.
\end{proof}

\noindent
\cref{prop2} is an analogy to the well-known result describing entanglement in graph states as a rank of submatrices of its adjacency matrix $\Gamma_G$, see \cref{rankkk} \cite{PhysRevA.69.062311}. We showed that this information is contained in the genarally smaller matrix $E_{\hat{G},T}$, with rows and columns indexed by the elements in the foliage partition $ \hat{V}$.

\subsection{k-uniform states}

A state $\ket{\psi}\in \mathcal{H}_2^{\otimes n}$ is \textit{$k$-uniform} if and only if it is maximally entangled with respect to any bipartition $A|A'$ where $|A|=k$, namely 
\[
\rho (\psi )_A \defeq\tr_{A'} \propto \text{Id},
\]
for any $|A|=k$ \cite{Helwig2013AbsolutelyME,helwig2013absolutely,
HelwigAME,raissi2019new,PhysRevA.100.032112,RaissiBurchardt22}. 
It is well-known that uniformity $k$ cannot exceed half of the dimension, \ie~$k\leq \frac{n}{2}$. 
Moreover if a state is $k$-uniform it is also $k'$-uniform for any $k'\leq k$ \cite{Helwig2013AbsolutelyME,helwig2013absolutely,Huber_2018}. A multipartite quantum state which achieves the highest possible uniformity, \ie~$k=\lfloor \frac{n}{2} \rfloor$ is called an \textit{Absolutely Maximally Entangled} (AME) state \cite{Helwig2013AbsolutelyME,helwig2013absolutely}. This special class of multipartite quantum states find applications in several quantum protocols including quantum secret sharing \cite{PhysRevA.59.1829}, parallel teleportation \cite{HelwigAME}, holographic quantum error correcting codes \cite{Pastawski2015HolographicQE}, and quantum repeaters \cite{AlsinaStab} among many others.

The notion of uniformity appears naturally in the context of graph and stabilizer states. There are many methods to construct graph states with high uniformity \cite{RaissiBurchardt22}. 
Note that a graph state $\ket{G}$ is $k$-uniform if and only if its entropy of entanglement $S_A (\ket{G})$ equals $k$ for any subset $A\subset [n]$ of size $|A|=k$. However, investigating the uniformity of a given graph state is a rather demanding problem and in principle requires analysis of an entire set of stabilizers (whose size is exponential in the number of qubits of the system). As an immediate consequence of \cref{rankkk}, we have the following.

\begin{observation}
\label{auxiA}
If graph $G$ is connected, the corresponding graph state $\ket{G}$ is 1-uniform.
\end{observation}

\noindent
As we will see now, the foliage partition provides a necessary and sufficient condition for $2$-uniformity of graph states. 

\begin{proposition}
\label{2UniCorrected}
Consider a connected graph $G$ associated with a graph state $\ket{G}$. For any two vertices $v,w\in V$, the corresponding two body marginal of $\ket{G}$ is proportional to the identity $\rho_{vw}\propto \Id$ if and only if the vertices $v,w$ belong to different parts in foliage partition, \ie~$v \in V_i$, $w\in V_j$, $i\neq j$, where $\hat{V}=\{V_1,\ldots,V_k\}$ is the foliage partition of $G$. 
\end{proposition}

\begin{proof}
Note that for a graph state $\ket{G}$ associated to a connected graph $G$, and any set $\{v,w\}$ consisting of two vertices, the entropy of entanglement $S_{\{v,w\}} (\ket{G}) = 2$ if the corresponding two-body marginal is maximally mixed, i.e $\rho_{vw}\propto \Id$, and $S_{\{v,w\}} (\ket{G}) = 1$ otherwise. 

Consider any set $V_i\in \hat{V}$ of the foliage partition of a graph $G$ with at least two parties $|V_i|\geq 2$, and any subset $W\subset V_i$ consisting of two elements, \ie~$|W|=2$. We shall see that the entropy of entanglement $S_{W} (\ket{G}) = 1$. 
Indeed, the graph $G$ can be transformed by LC-operations to the form $G'$, where all vertices in $V_i$ have pairwise the same neighborhood. Such a graph has the same entanglement properties, hence without loss of generality, we may assume that all vertices in $V_i$ have pairwise the same neighborhood. 
Consider the bipartition $W |W'$ and the corresponding matrix $\Gamma_G^{W\,W'}$. 
Note that all rows in the matrix $\Gamma_G^{W\,W'}$ are the same, hence $\text{rank}\; \Gamma_G^{W\,W'} \leq1$. 
According to \cref{auxiA}, and \cref{rankkk}, the entropy of entanglement $S_{W} (\ket{G}) = 1$. Therefore the state $\ket{G}$ cannot be $2$-uniform. 

On the other hand, suppose that there is a two-element subset $T=\{v_1,v_2\}$, such that $\rho_T (\psi) \not\propto  \text{Id}$. Note that this means $S_{T} (\ket{G}) <2$, and since entropy of entanglement takes integer values for graph states, by \cref{auxiA}, $S_{T} (\ket{G}) =1$. Consider now a bipartition $T|T'$ and the corresponding matrix $\Gamma_G^{T \,T'}$. Denote by $r_1, r_2$ rows in $\Gamma_G^{T \,T'}$  corresponding to $v_1$, and $v_2$ vertices respectively. By (\ref{rankkk}), $\text{rank}\,\Gamma_G^{T \,T'}=1$, hence either $r_1=r_2$, or one of two vectors $r_1$ or $r_2$ is the zero vector. In the first case, $v_1$ and $v_2$ have pairwise the same neighborhood, in the second case, vertices $v_1$ and $v_2$ are leaf and adjacent axil. In both cases they belong to the same set in the foliage partition, hence, the foliage partition cannot be trivial, \ie~$|\hat{V}|\neq |V|$.
\end{proof}

\noindent
\new{\Cref{2UniCorrected} establishes an intimate relationship between the foliage partition of a graph and the $2$-body marginals of the corresponding graph state. Specifically, one uniquely determines the other. This relationship suggests potential connections between the foliage partition and certain entanglement measures.

Firstly, it is notable that the foliage partition is closely related to the concept of local sets introduced in Ref.~\cite{LocalSets1}. In particular, any pair of distinct vertices $u, v$ belong to the same set in the foliage partition if and only if ${u,v}$ is a minimal local set\footnote{A set $M$ is termed local if it has the form $M = D \cup \Odd{D}$, where $\Odd{D} = \{v \in [n] : |\delta_v \cap D| \text{ is odd}\}$ is the set of vertices connected to an odd number of vertices in $D$. Here, $\delta_v$ denotes the neighborhood of vertex $v$. A local set that is minimal by inclusion is called a minimal local set.} \cite{claudet2023small,claudet2024covering}.

Secondly, it is important to recognize that the foliage partition does not convey any information about higher-order marginals (i.e.~not about 3-body marginals or greater). 
Therefore, it is inherently impossible to establish a direct link between the foliage partition and other entanglement measures that are based on higher-order correlations, such as entanglement width \cite{PhysRevA.75.012337,PhysRevLett.97.150504,PhysRevLett.131.030601} or sector length \cite{Wyderka_2020}. However, this does not exclude the possibility that these types of entanglement measures may be related to the recursive application of the foliage partition, i.e., the $k$-th foliage graph, as discussed in \cref{FoliageRepresentation}.
}

As a consequence of \cref{2UniCorrected}, the foliage partition provides a necessary and sufficient condition for $2$-uniformity of graph states. 

\begin{corollary}
\label{2Uni}
For a connected graph $G$, the associated graph state $\ket{G}$ is $2$-uniform if and only if the graph's foliage partition $\hat{V}$ is trivial, \ie~$|\hat{V}|=|V|$. 
\end{corollary}

\noindent
The reminiscence of the above relation between $2$-body marginals and a form of foliage partition of a graph state can be found in \cite{gittsovichMultiparticleCovarianceMatrices2010,Wyderka_2020} %\cite{OGRem} 
 where the relation between the $2$-body marginals and the existence of axils, leaves, and twin vertices was described.

\cref{2Uni} shows that the property of being a $2$-uniform graph state is determined by the graph's foliage partition. 
It is intriguing if there is another LC-invariant for graphs that determines higher uniformities of corresponding quantum states. We leave this for future work. 

\section{LC-classes of graph states}
\label{Section:asymptotic}

\new{
We say that two graphs $G$ and $G'$ belong to the same \textit{LC-class} if $G$ is LC-equivalent to $\sigma(G')$, where $\sigma$ is some permutation of nodes. In other words, the LC-class of a graph $G$ is the set of all unlabeled graphs that can be obtained by performing any sequence of LC operations on $G$. We denote by $\mathcal{C}(n)$ the number of LC-classes among all graph states with $n$ qubits. The classification of graph states into LC-classes is a well-studied topic; in particular, graph states up to $12$ qubits were classified into $\sim 1.3 \times 10^6$ different LC-classes \cite{Danielsen_2006,PhysRevA.83.042314}.

As we shall see, the foliage partition provides a lower bound on the number of LC-classes $\mathcal{C}(n)$, which is exponential in $\sqrt{n}$. 

The foliage partition $\hat{V} = \{V_1, \ldots, V_k\}$ of a graph $G$ directly provides a partition $|V_1|, \ldots, |V_k| \mathbin{\vdash} n$ of an integer $n$. Clearly, if for two graphs $G$ and $G'$, the corresponding integer partitions are different, $G$ and $G'$ belong to different LC-orbits. Conversely, for almost all integer partitions of a given $n$, one can find a graph with a foliage partition corresponding to this integer partition. Indeed, we have the following.

\begin{proposition}
\label{prop:Int_partition}
For non-decreasing integer partitions $n_1, \ldots, n_k \mathbin{\vdash} n$ (where $n_{i-1} \leq n_i$), there exists a graph $G$ with a foliage partition matching this partition, except in the cases $1, n-1 \mathbin{\vdash} n$; $1, 1, n-2 \mathbin{\vdash} n$; and $1, 1, 1, n-3 \mathbin{\vdash} n$.
\end{proposition}

\begin{proof}
This correspondence follows from direct calculation and the consideration of several cases. 

Firstly, note that the cyclic graph on five or more nodes has a trivial foliage partition. Therefore, when $k \geq 5$, the foliage partition of a graph $G$ of the type presented in \cref{fig:partition} always corresponds to the integer partition $n_1, \ldots, n_k \mathbin{\vdash} n$. 

Secondly, when $k \leq 4$, assume first that $n_{k-1} \geq 2$ (and hence $n_k \geq 2$ due to the non-decreasing condition). A straightforward analysis of the three separate cases $k = 2, 3, 4$ shows that the foliage partition of $G$ presented in \cref{fig:partition} is indeed $n_1, \ldots, n_k \mathbin{\vdash} n$.

On the other hand, if the integer partition satisfies $k \leq 4$ and $n_{k-1} = 1$, this means that it is either $1, n-1 \mathbin{\vdash} n$, $1, 1, n-2 \mathbin{\vdash} n$, or $1, 1, 1, n-3 \mathbin{\vdash} n$. In each of these three cases, it is easy to see that there is no graph with a corresponding foliage partition.
\end{proof}

\noindent
\cref{prop:Int_partition} provides the following estimation of the number of LC-classes $\mathcal{C}(n)$.

\begin{corollary}
\label{cor:partition}
The number of of LC-classes $\mathcal{C}(n)$ is bounded by
\begin{equation}
\label{eq:clas_vs_part}
\mathcal{C}(n)\geq
p(n)- \text{min}(3, n-2)
\end{equation}
where $p(n)$ is the number of distinct ways to represent $n$ as a sum of positive integers. 
\end{corollary}

\noindent
The function $p(n)$ is called the \textit{partition function}, and it is a well-known function in number theory. Despite the low values of the function $p(n)$ for small values of $n$, it grows exponentially in $\sqrt{n}$. Indeed, $p(n)$ has the following asymptotic behavior:
\begin{equation}
{\displaystyle p(n) \sim {\frac {1}{4n{\sqrt {3}}}} \exp \left( \pi {\sqrt {\frac {2n}{3}}} \right)},
\end{equation}
which was shown by Hardy and Ramanujan \cite{HardyRama1918}. For instance, for $n = 100$, the partition function is $p(n) \sim 1.9 \times 10^9$, which directly translates into the lower bound of the number of LC-classes $\mathcal{C}(n)$ in accordance with (\ref{eq:clas_vs_part}). \Cref{tab:assymptotic} presents such a lower bound on $\mathcal{C}(n)$ given by the foliage partition for small values of $n$.

We end this section with the following remark, comparing the estimation of $\mathcal{C}(n)$ with respect to the foliage of a graph and the foliage partition. Note that the foliage of a graph provides a bipartition of an integer $n = n_\text{f} + n_\text{nf}$, where $n_\text{f}$ is the number of vertices in the foliage and $n_\text{nf}$ is the number of vertices outside the foliage. The number of such bipartitions is obviously $n + 1$. It is easy to see that for $n \geq 7$ there are graphs with foliage corresponding to every bipartition, except the one for which $n_\text{f} = 1$. Therefore, for $n \geq 7$, the lower bound on the number $\mathcal{C}(n)$ resulting from the foliage of a graph is $\mathcal{C}(n) \geq n$, and it can be analyzed separately for $n < 7$, see \Cref{tab:assymptotic}. Note that the foliage of a graph provides a lower bound on $\mathcal{C}(n)$ which is linear in $n$, while the bound following from the foliage partition is exponential in $\sqrt{n}$.
}

\begin{figure}[ht!]
\begin{center}
\includegraphics[width = 0.4\textwidth]{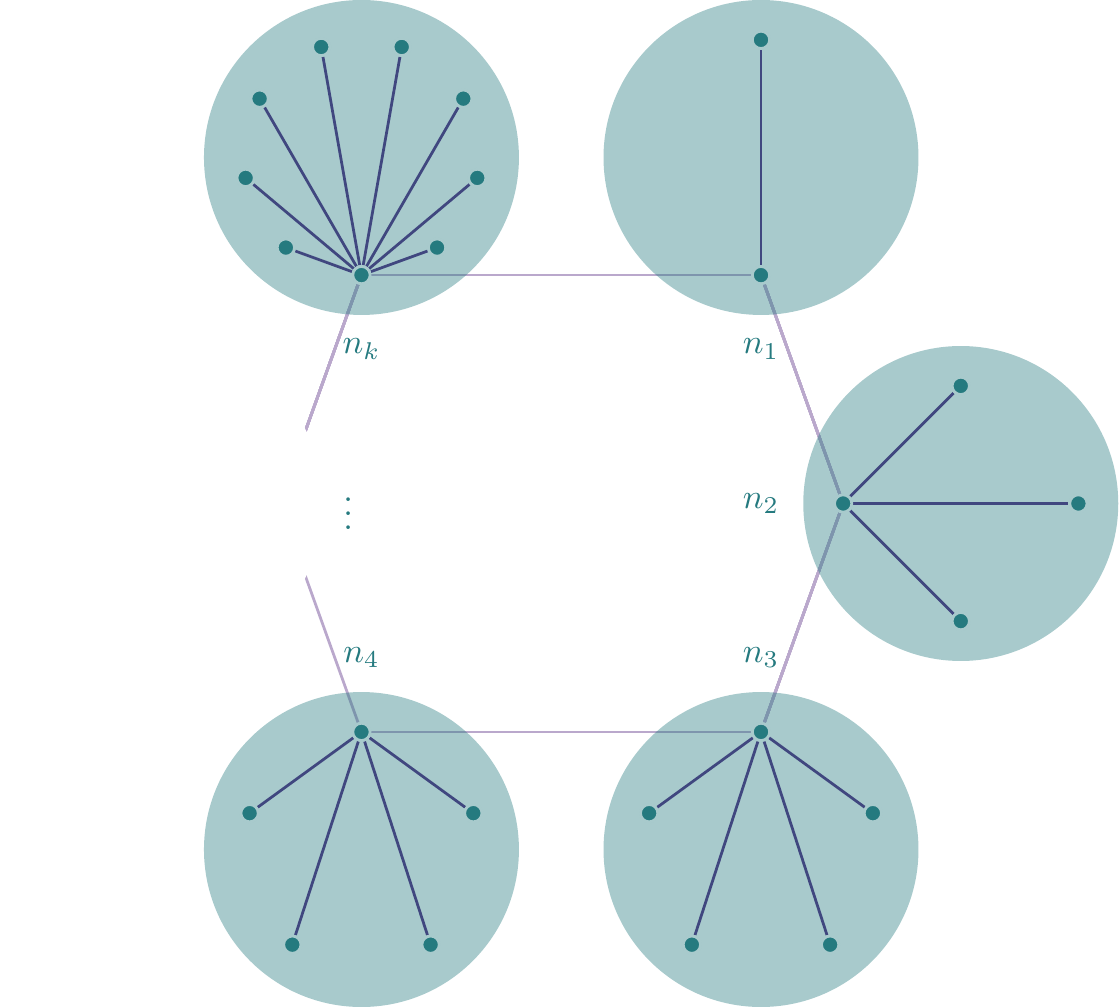}
\end{center}
\caption{\label{fig:partition}
\new{Graph $G$ constitutes of $k$ trees, with $n_k$ nodes each connected together in a cyclic way. Except of three special cases, the corresponding foliage partition provides an integer partition $n_1,\ldots ,n_k \mathbin{\vdash} n$.}
}
\end{figure}

\begin{table}[ht]
\begin{center}
% \begin{tabular}{|c|| c |c| c|} 
%  \hline
%   & & $\#$ of LC-classes &  $\#$ of LC-classes  \\%&  $\#$ of non-\\ 
% $n$ & $\#$ of LC-classes  &distinguished by & distinguished by   \\%&isomorphic\\
% &  & foliage partition & foliage \\%& trees\\
% [0.5ex] 
%  \hline
%  $\leq $5 & 8 & 8&7 \\%&7\\
% 6 & 11 & 8& 5 \\%&6\\
% 7 & 26 & 12&7 \\%&11\\
% 8 & 101 & 19&8 \\%&23\\
% 9 & 440 & 27&9 \\%&47\\
% 10 & 3 132 & 39&10 \\%&106\\
% 11 & 40 457 & 53&11 \\%&235\\
% 12 & 1 274 068 & 74&12 \\%&551\\
%  \hline
% \end{tabular}
		\includegraphics[width=0.55\columnwidth, 
    trim=0cm % left
        0cm % bottom
        0cm % right
        0cm,% top
    clip]{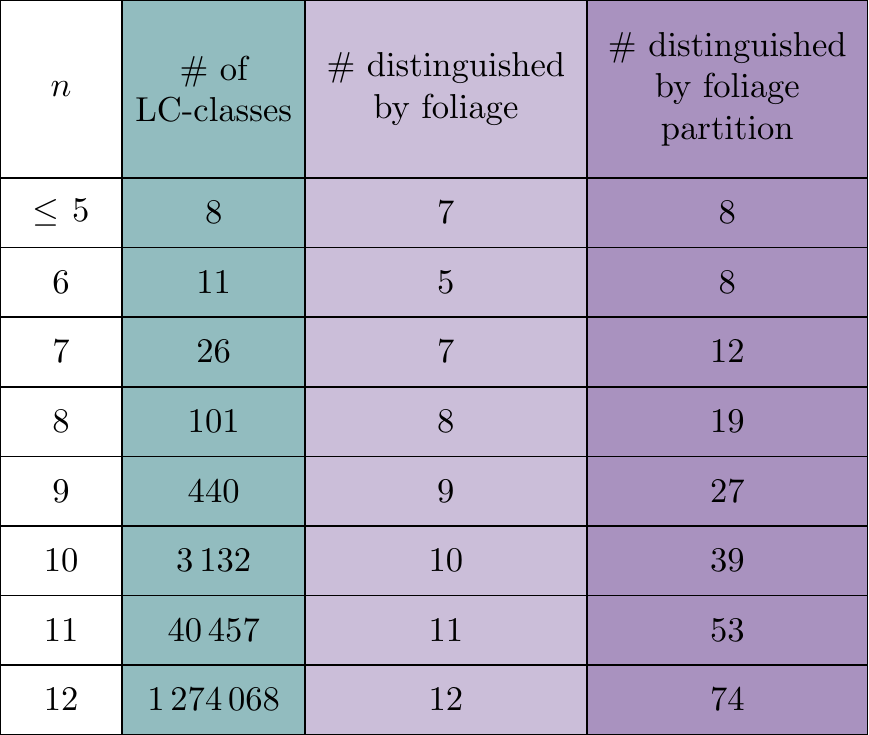}
    \caption{
    \textcolor{green}{The number of LC-classes for $n\leq 12$ qubits}, and lower bounds on this number \textcolor{violet}{derived from the foliage partition and only the foliage of a graph}, respectively. 
    % Despite the relatively small numbers in the fourth column, this bound is exponential in $\sqrt{n}$. By comparison, the bound in the third column is linear in $n$. 
    The bound in the third column is linear in $n$, while the one in the fourth column, is exponential in $\sqrt{n}$. 
    }
    \label{tab:assymptotic}
    \end{center}
\end{table}

\section{LC-symmetries of a graph}
\label{Section:symmetry_of_graphs}

In this section, we discuss the \textit{LC-automorphism group} of a graph, \ie~the group of vertex permutations  that can only produce graphs that are LC-equivalent to the original graph.
In particular, we discuss the relation between the LC-automorphism group and the foliage partition of a graph. 

Recall that the automorphism group of a graph is the group of all permutations $\sigma \in S_n$ of its vertices such that any pair of vertices $(u,v)$ forms an edge if and only if the pair $(\sigma (u),\sigma  (v))$ also forms an edge. Determining the group of automorphisms of a graph is computationally hard. In particular, the \textit{graph automorphism problem}, \ie~the problem of deciding whether a graph has a non-trivial automorphism or not, is in the computational complexity class 
\texttt{NP}~\cite{LUKS198242,MATHON1979131,doi:10.1137/0210002}. 

A similar problem, which arises naturally in quantum computation and information, is to determine the group of permutations of graph vertices under which corresponding graph states remain locally Clifford (or unitary, or SLOCC) equivalent. 
Surprisingly, the problem has been the subject of very little research to date~\cite{BurchardtRaissi20,burchardt2021entanglement,Balasz22XYZ,Rather22classLU}. 
Formally, we define the following.

\begin{definition}
\label{LCautoGroup}
Consider a graph $G=(V,E)$ and a permutation $\sigma$ of its vertex set $V$, \ie~$\sigma (G)$ is the graph with its vertices and edges permuted according to $\sigma$. 
If $\sigma(G)$ is LC-equivalent to $G$, then we call $\sigma$ an \textit{LC-automorphism} of $G$. The set of all LC-automorphisms forms the \textit{LC-automorphism group} denoted by $\AutG$.
\end{definition}

\noindent
Note that the LC-automorphism group is, indeed, a group as the composition of two LC-operations, and the inverse of local complementations are local complementations. As we will see, the foliage partition of a graph is a relevant notion in order to determine the LC-automorphism group of a given graph.

\begin{lemma}
\label{symmetries_aut}
Let $\{V_1,\ldots,V_k\}=\hat{V}$ be the foliage partition of a graph $G$. The following subset 
\begin{equation}
\label{eqq9}
\AutGin 
\defeq\Big\{ \sigma_1 \cdots\sigma_k \Big| \sigma_i\in S_{V_i} \Big\} 
= S_{V_1} \times\cdots \times S_{V_k}
\end{equation}
is a normal subgroup of the group of all LC-automorphisms of $G$, \ie~$\AutGin\triangleleft\AutG$.
\end{lemma}

\begin{proof}
Consider any subset $V_i$ of the foliage partition of a graph $G$. The graph $G$ can be transformed by LC-operations to the form, where all vertices in $V_i$ have pairwise the same neighborhood\footnote{This can be seen from the transformation rules of \cref{LCinLLC}. All sets of type \axilleaf~can be transformed into \contwins~by a local complementation with respect to the axils. The neighboring sets of type \discontwins~and \contwins~are only transformed into each other by this.}. Such a graph is invariant under any permutation of vertices $\sigma_i \in S_{V_i}$. Therefore the state $\ket{\psi}_G$ corresponding to the initial graph $G$, and the state with permuted parties $\id \otimes \cdots \otimes  \sigma_i \otimes \cdots \otimes \id \ket{\psi}_G$ are LC-equivalent. Since LC-equivalence is a transitive relation, we can compose any such permutations of subsets of the foliage partition and therefore that $\AutGin$ is a subset of $\AutG$. The set $\AutGin$ is closed under compositions and inverses and is hence a subgroup of $\AutG$.

Secondly, we shall show that $\AutGin$ is a normal subgroup of $\AutG$. 
We shall prove that it is normal. Suppose that $\sigma\in \AutG$ and $\pi\in \AutGin$. We shall show that $\sigma \circ\pi \circ \sigma^{-1} \in \AutGin$. 
Suppose on the contrary, that $\sigma \circ\pi \circ \sigma^{-1} \not\in \AutGin$, meaning that there is a vertex $v\in V_\ell$, such that $\sigma \circ\pi \circ \sigma^{-1} (v) \not\in V_\ell$. Denote by $V_{\ell '} \in \hat{V}$ the set in the foliage partition which contains $\sigma (v)$, i.e $\sigma (v) \in V_{\ell '}$. Since $\pi \in \AutGin$, $\pi \circ \sigma (v) \in V_{\ell '}$ belongs to the same set in foliage partition. Therefore both vertices are transformed by $\sigma^{-1} \in \AutG$ into vertices in the same set in the foliage partition. Note, however, that $v=\sigma^{-1}\circ \sigma (v) \in V_{\ell }$, while $\sigma^{-1}\circ \pi \circ \sigma (v) \not\in V_\ell$,  which leads to the contradiction. This shows that $\AutGin$ is a normal subgroup of $\AutG$, and finishes the proof.
\end{proof}

Since $\AutGin \triangleleft \AutG$ is a normal subgroup, the group $\AutG$ can be represented as a semi-direct product\footnote{For a given group $G$ and its normal subgroup $N$ and subgroup $H$, we say that the group $G$ is a \textit{semi-direct product} $G =  N \rtimes H $ if and only if for any element $g\in G$ there are unique $n\in N$ and $h\in H$ such that $g=nh$. The group $H $ is isomorphic to the quotient group $G  / N$ and the orders of three groups satisfy $|G|=|N||H|$. In that sense, the semi-direct product is a generalization of a direct product, where only one subgroup is normal. Similarly as with direct products, $G$ decomposes uniquely into $N$ and $H$, which however, do not commute (contrary to direct product).} of $\AutGin$ and the quotient group, that is, as $\AutG =  \AutGin  \rtimes \AutG  /\AutGin  $. Furthermore $\AutG  /\AutGin $ can be seen as a subgroup of permutations of subsets of foliage partition.

\begin{lemma}
\label{lemm2}
Let $\hat{V}=\{V_1,\ldots,V_k\}$ be the foliage partition of a graph $G$, and let $v_i\in V_i$ be representatives of each partition set. For any permutation $\sigma \in S_{V}$ of the vertices, define the permutation $F(\sigma ) \in S_{\hat{V}}$ of the foliage partition by $F(\sigma ) :V_i\mapsto V_{j}$ where $j$ is such that $\sigma (v_i)\in V_{j}$.  The quotient group $ \AutG  /\AutGin  $ is isomorphic to the following subset of permutations of subsets of the foliage partition
\begin{align*}
\AutGout \cong \Big\{F(\sigma ) \in S_{\hat{V}}
\big| 
\sigma \in \AutG, \Big\}
<S_{\hat{V}}.
\end{align*}
\end{lemma}

\begin{proof}
Firstly, note that any automorphism $\sigma \in \AutG$ induces a permutation $\hat{\sigma} \in S_{\{V_1,\ldots,V_k\}}$ on the partition set. Indeed, for any $i\in [k]$, choose a representative element $v_i\in V_i$. For any $i\in [k]$, there exists $i'\in [k]$ such that $\sigma (i) \in V_{i'}$, which defines $\hat{\sigma} $ by $\hat{\sigma} (V_i )=V_{i'}$.

We shall see that this is a well-defined permutation of the foliage partition. Indeed, as an immediate consequence of \cref{necessary_condition}, any two vertices $v_i, w_i \in V_i$ from the same set in foliage partition are transformed by $\sigma $ onto vertices belonging to the same set $\sigma (v_i), \sigma (w_i) \in V_{i'}$, hence $\hat{\sigma}$ does not depend on the choice of representatives.  Furthermore, it is easy to see that $\hat{\sigma}$ is a surjection and hence permutation of a foliage partition. 

Observe that the set $ \AutG $ constitutes a group with respect to composition of permutations. Furthermore, the mapping $F:\AutG \rightarrow \AutGout$ given by $F(\sigma)=\hat{\sigma}$ is a group homomorphism. Note that the group $\AutGout$ was defined as an image of $F$, \ie~$\text{Im} F =\AutGout$. We shall see that the kernel is given by $\text{Ker} F =\AutGin$. Observe that from $\hat{\sigma}\equiv \text{id}$ it follows that for all $i\in [k]$ there are $v_i\in V_i$ such that $\sigma (v_i)\in V_i$. Furthermore, for all $i\in [k]$ and all $v\in V_i$, $\sigma (v)\in V_i$, hence $\sigma \in \AutGin$. On the other hand, for any $\sigma\in \AutGin$, $F(\sigma)\equiv \text{id}$ on the foliage partition.

Since $F$ is a group homomorphism it lifts into an isomorphism $\hat{F}: \AutG / \text{Ker} F \rightarrow \text{Im} F =\AutGout$, and as we have shown $\text{Ker} F=\AutGin$, which finishes the proof.
\end{proof}

The foliage partition of a graph does not completely characterize the group $\AutGout$, however, it restricts its structure in the following way. 

\begin{lemma}
\label{lemm1}
Let $\hat{V}=\{V_1,\ldots,V_k\}$ be the foliage partition of a graph $G$. We have the following bound
\begin{align}
\label{eqq10}
\AutGout  < \AutGoutupp \defeq
S_{\hat{V}^1} \times \cdots  \times S_{\hat{V}^k} ,
\end{align}
where $\hat{V}^{\ell}=\{ V_i\in \hat{V}: |V_i|=\ell\}$ denote subset containing all sets with exactly $\ell$ elements.
\end{lemma}

\begin{proof}
In accordance to \cref{necessary_condition}, any two vertices $v_i, w_i \in V_i$ from the same set in foliage partition are transformed by $\sigma $ onto vertices belonging to the same set $\sigma (v_i), \sigma (w_i) \in V_{i'}$. As a consequence, sets $V_i$ and $V_{i'}$ has the same order, and hence $F(\sigma)$ permutes only sets in the foliage partition of the same order
\end{proof}

\cref{Table11} presents LC-automorphisms of several graphs on five qubits, previously presented in \cref{figPartition}.

\begin{table}[ht!]
\begin{center}
% \begin{tabular}{|c|| c |c| c|c|c |c|} 
%  \hline
%   &$\AutGin $& $|\AutGin |$ & $\AutGout$ & $|\AutGout |$ &$|\AutGoutupp |$& $\AutG$ \\ %[0.5ex] 
%  \hline\hline
% $K_5$ & $S_5$ & 120 & $\Id$ & 1 &1 & 120 \\ 
%  \hline
% $K_{2,3}$ & $S_2\times S_3$ &12& $\Id$ & 1 &1 & 12 \\
%  \hline
% $C_{5}$ & $\Id$ &1& $S_5$ & 120 &120 & 120 \\
%  \hline
% \end{tabular}
		\includegraphics[width=0.55\columnwidth, 
    trim=0cm % left
        0cm % bottom
        0cm % right
        0cm,% top
    clip]{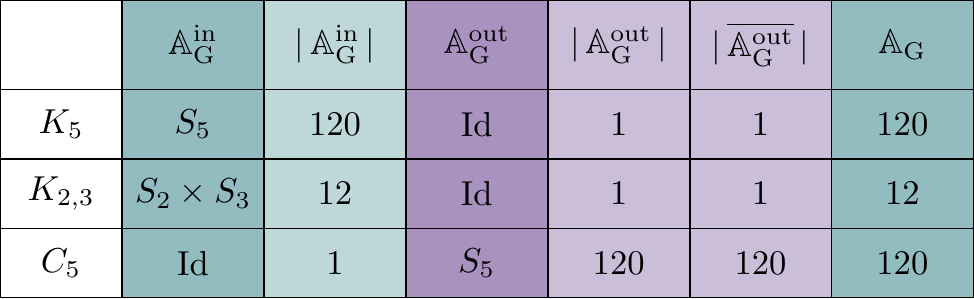}
\caption{\label{Table11} The LC-automorphisms group $\AutGin $ of the three graphs previously presented in \cref{figPartition}:
 the fully connected graph $K_5$, the bipartite graph $K_{2,3}$, and the circular graph $C_5$. As $|\AutGoutupp |=1$ for the first two graphs, the foliage representation completely characterizes their LC-automorphisms group $\AutGin $.}
\end{center}
\end{table}

As we demonstrated, the group $\AutG$ of all LC-automorphisms of a graph is a semi-direct product of two groups $\AutG =  \AutGin  \rtimes \AutGout   $, where $\AutGin $ consists of all permutations within subsets of the foliage partition, \ie~$\AutGin =S_{V_1} \times\cdots \times S_{V_k}$ while $\AutGout$ is isomorphic to a subgroup of all permutation of the subsets in foliage partition, \ie~$\AutGout < S_{\{V_1,\ldots,V_k\}}$. Note that the group $\AutGin$ is completely characterized by the foliage partition, and its order is $|\AutGin | =\prod_{i=1}^k |V_i | !$. Furthermore the foliage partition provides restriction on the structure and order of $\AutGout$. Directly form \cref{lemm2,lemm1}, we have the following.

\begin{proposition}
\label{symmetriess}
The LC-automorphisms group $\AutG$ of a graph $G$ satisfies
\begin{equation*}
    \AutGin \,\leq \, \AutG \,\leq \,
 \AutGin   \rtimes \AutGoutupp,
\end{equation*}
where $ \AutGin$ and $\AutGoutupp$ are completely characterized by the foliage partition $\hat{V}=\{V_1,\ldots,V_k\}$ of a graph, see \cref{eqq9,eqq10}. Furthermore, denote by $\hat{V}^{\ell}=\{ V_i\in \hat{V}: |V_i|=\ell\}$ subset containing all sets with exactly $\ell$ elements, and by $T_\ell=|\hat{V}^{\ell}|$ its size. Then, the order of the LC-automorphisms group is bounded as follows:
\begin{align*}
\prod_{\ell =1}^k \, {\ell !}^{T_\ell} \,
%=\Big| \AutGin  \Big|
\leq \Big| \AutG  \Big| \leq 
%\Big| \AutGin  \Big| \cdot \Big| \AutGoutupp  \Big| =
\prod_{\ell =1}^k \, {\ell !}^{T_\ell} \, T_\ell !
\end{align*}
\end{proposition}

%\noindent
Note that the bound given above is tight if $| \AutGoutupp |=1$, which occurs if all sets in the foliage partition have a different number of elements, \eg~see the first two graphs in \cref{Table11}.
We analyzed the $\AutG $ group of all LC-orbits of graph states up to 8 qubits. In \cref{tab:symmetries_table1,tab:symmetries_table2,tab:symmetries_table3}, we present the LC-automorphism groups of all LC-orbits up to eight qubits and show when the above lower and upper bounds are saturated.

We conclude this section by connecting the size of LC-automorphisms group $\AutG$ with the size of LC-orbit of $G$. Recall that LC-orbit $L_G$ consists of all graphs LC-equivalent to $G$. Furthermore, identifying all isomorphic graphs in $L_G$ results into orbit $C_G$, relevant in the context of unlabelled graphs. Formally, $C_G\defeq L_G / \sim$, where two graphs $G,G'$ are related $G\sim G'$ if they are isomorphic, i.e there is $\sigma\in S_n$ such that $G=\sigma (G')$. 

\begin{proposition}
The size of LC-automorphisms group of a graph and its $L_G$ and $C_G$ orbits satisfy the following inequality:
\begin{equation}
\label{w23}
|\AutG | \; |C_G | \geq | L_G |.
\end{equation}
\end{proposition}

\begin{proof}
Consider an equivalence class in $C_G$ orbit represented by a graph $G'\in L_G$, we denote it $[G']$. Note that $[G']=\{ \sigma (G') :\sigma\in \AutG \}$. Obviously, the size of $[G']$ orbit is at most $|\AutG |$\footnote{In fact, the size of the orbit $[G']= |\AutG | / |A_G|$, where $A_G $ is the automorphisms group of a graph $G$.}. Since $C_G$ consists of orbits in $L_G$ and each orbit contains at most $|\AutG |$ elements from $L_G$, the inequality (\ref{w23}) holds true. 
\end{proof}

This observation can be used to bound the size of each of the three quantities involved. For example, the ratio $\tfrac{|L_G |}{|C_G |}$ gives the lower bound on the LC-automorphisms group of a graph. Furthermore, the quantity $I_G \defeq  |\AutG |\, \tfrac{|L_G |}{|C_G |}$ satisfies $I_G\geq 1$ and measures the average size of the automorphisms group of the graphs in $C_G$ orbit. Analizing all LC-orbits of graphs up to $n=8$ qubits, we noticed that the quantity $I_G$ is closely related to the structure of foliage partition, see the last column in \cref{tab:symmetries_table1,tab:symmetries_table2,tab:symmetries_table3}. Indeed, it is close to its lower bound $I_G \geq 1$ for graphs having the trivial foliage partition, i.e $|V|=|\hat{V}|$, meaning that all LC-equivalent graphs have small number of automorphisms in average. On the other hand $I_G$ is maximized for the orbit of GHZ state, the only states for which $|\hat{V}|=1$. Indeed, all LC-equivalent states are either fully-connected graph or a star graph, both with a large automorphisms group.

\section{An algorithm to find the foliage partition}
\label{algo}

Based on \cref{def:foliage_partition}, we can give an algorithm to find a foliage partition of a given graph that runs in $\mathcal{O}(n^3)$ time in the number of graph vertices $n$. 
Indeed, consider a graph $G$ represented by its adjacency matrix $\Gamma_{G}$ and denote by $\Gamma_{v}$ the $v$-row of the $\Gamma_{G}$ matrix. 
In each step, the algorithm takes a vertex $v\in V$ as an input and returns a subset $V_i\in \hat{V}$ of the foliage partition of all vertices related to $v$. We save this information and delete the subset $V_i$ from the graph. We continue this process until the graph is empty. Below, we present the iteratve step of the algorithm for a vertex $v\in V$ as an input.

Firstly, check if $v$ is a leaf. Note that $v$ is a leaf if and only if $\Gamma_{v}$ contains only one non-zero element. If such a unique non-zero element was in $w$ position, $w$ is an adjacent axil. Hence $\Gamma_{w}$ contains information of all other leaves adjacent to $w$. The set of those leaves together with $w$ constitutes a subset $V_i$ in the foliage partition. The complexity cost is $\mathcal{O}(n^2)$, since for each of the neighbors of $w$ we need to check if they are leafs.

Secondly, if $v$ is not a leaf, check if it is an axil. This can be done by choosing any neighbor of $v$ and checking if it is a leaf. If this is the case, $v$ together with all its neighbors (read $\Gamma_{v}$) constitutes a subset $V_i$ in the foliage partition. The complexity cost is again $\mathcal{O}(n^2)$. 

Finally, if $v$ is neither a leaf nor an axil, we shall determine all twins sharing the same neighborhood with $v$. Choose any other vertex $w\in V$ and decide if $v$ and $w$ share the same neighborhood. For instance, compare $\Gamma_{v}$ against $\Gamma_{w}$, both agree on all positions except $v$ and $w$ if and only if they share the same neighborhood. Repeat this for every vertex $w\in V$ in the graph and store vertices sharing the same neighbors as $v$. Those vertices together with $v$ constitute a subset $V_i$ in a foliage partition. The complexity cost is $\mathcal{O}(n^2)$.

Since the computational complexity in each step is at most $\mathcal{O}(n^2)$, and we repeat this process at most $n$ times, the total computational complexity is $\mathcal{O}(n^3)$. Hence, we have the following.

\begin{lemma}[An algorithm to find foliage partition]
\label{algoAlgo}
There is an algorithm to find a foliage partition of a given graph that runs in $\mathcal{O}(n^3)$ time in the number of vertices $n$ in the graph.
\end{lemma}

Interestingly, the computational complexity of an algorithm providing a foliage partition of a graph based on the relation defined \cref{def:relation} is larger, namely $\mathcal{O}(n^4)$. Indeed, the relation presented in \cref{def:relation} involves four vertices at once and hence its straightforward implementation has complexity $\mathcal{O}(n^4)$ in the number of vertices of a graph $n$.

\new{Note that the foliage of a graph is efficiently computable, as a direct consequence of this result. Indeed, the foliage of a graph is defined as the union of all elements of the foliage partition that are of size two or greater. Therefore, by computing the foliage partition, we can also determine the foliage of a graph. Since our algorithm computes the foliage partition in \( O(n^3) \) time, the foliage can likewise be computed in \( O(n^3) \). }

\section{Foliage partition for small graphs}
\label{Section:foliage_of_small_graphs}

Using the existing repositories and libraries of representatives of LC-orbits of graph states on $n\leq 8$ qubit subsystems \cite{Adcock2020mappinggraphstate,Table_LU_orbit,PhysRevA.83.042314}, we computed the foliage partition and representation for all LC-orbits of small graphs. \cref{fig:graph_representatives} presents foliage partition of all LC-orbits of graphs up to $n=8$ vertices. There are exactly 146 such orbits, the order in which they are listed and selection of orbits representatives agree with Figure 11 in Ref. \cite{Adcock2020mappinggraphstate}. Vertices belonging to the same set in the foliage partition are presented with the same colour. 

Furthermore, \cref{fig:graph_representatives_level1} presents the foliage graph of all graphs in \cref{fig:graph_representatives}, where nodes in the same set in the foliage partition were grouped together. Note that graphs with the same (or more generally LC-equivalent) foliage graph does not have to be equivalent. Indeed, representatives (13) and (15) have the same foliage graph, however are not LC-equivalent, as they represents different LC-orbits, see \cref{fig:graph_representatives,fig:graph_representatives_level1}. Note that they are distinguished by the type function, which takes values \axilleaf, \axilleaf, \axilleaf~for representative (13) and \axilleaf, \axilleaf, \contwins~for representative (15).

In addition, \cref{tab:symmetries_table1,tab:symmetries_table2,tab:symmetries_table3} presents general results concerning LC-automorphisms group $\AutG$ of representatives of all LC-orbits of graphs up to $n=8$ vertices. As we have shown, $\AutG =  \AutGin  \rtimes \AutGout   $ is a semi-direct product of two groups, where $\AutGin$ is completely determined by the foliage partition, while $\AutGout$ has usually more sophisticated structure. As a consequence, the form and the size of $\AutG$ group is only bounded by the foliage partition, see \cref{symmetriess}. Cases for which the bound is tight (lower bound equals upper bound in \cref{symmetriess}) are shown in brown, otherwise, cases where the lower/upper bound is achieved are presented in green/purple. Using \texttt{GAP} programming, we computed the form and generators of $\AutGout$ group for representatives of all LC-orbits. As generators of $\AutGin$ group are fully characterized by the foliage partition, see \cref{symmetries_aut}, this gives a full insight for the form of $\AutG$ group. 

We conclude this section with two loose observations relating to the foliage partition and the group of LC-automorphisms $\AutG$ group. Firstly, note that the $2$nd-foliage partition (and in general $k$th foliage partition) might give better upper bounds on the $\AutGout$ group. Indeed, it is easy to see, that $\AutGout$ is not only a subgroup of $\AutGoutupp$ but also $\overline{\mathbb{A}_{\hat{G}}^{ \text{out}}}$, and $\overline{\mathbb{A}_{\hathat{G}}^{ \text{out}}}$, etc. In particular, for the representative of (139) LC-orbit $G$, presented on \cref{fig:graph139}, we have $\AutGoutupp \cong S_6$, while $\overline{\mathbb{A}_{\hat{G}}^{ \text{out}}}\cong S_5$, which actually agree with the $\AutGout \cong S_5$. This phenomenon might be useful when examining LC-automorphisms $\AutG$ group of larger graphs. 

Secondly, note that the property of being a subgraph has no immediate consequences connecting related LC-automorphisms groups. Indeed, for $H$ being a subgraph of $G$, we encounter many examples when $\AutG  < \AutH $ (\eg~orbits 47, 8), $\AutH< \AutG$ %and the foliage of $G$ is smaller/the same as the foliage of $H$ 
(\eg~orbits 4, 120, or orbits 43, 141). This loosely resembles a decision problem of whether the graph H can be obtained from the graph G by sequences of LC-transforms and vertices deletaion, which is known to be \texttt{NP}-hard \cite{Dahlberg_2020,Dahlberg2020transforminggraph}. It is intriguing to ask if there is any relation between the inclusion of groups of automorphisms of a graph and its subgraph and the problem of local transformation between the corresponding states described in Ref. \cite{Dahlberg_2020,Dahlberg2020transforminggraph}.

\section{Qudit graph states}
\label{Section:qudits}

In this section we recall the definition of a qudit graph states $\ket{G}\in \mathcal{H}^{\otimes n}_d$	and their relation to weighted graphs. In the special case, when $d=2$, they coincide with the graph states presented in \cref{Section:qubits}. All results presented in the paper, as foliage partition, foliage representation, foliage graph, might be easily adapted to the qudit graph states as we present in this section.

\textit{Weighted graphs} are a natural generalization of graphs, where edges are equipped with weights, here taken from a finite group $\mathbb{Z}_d$. A $d$-weighted graph is a simple graph $G=(V,E)$  with weights $ \omega_{vw} \in \mathbb{Z}_{d}$ on its edges $vw \in E$. By convention, $vw \not\in E$ if and only if the corresponding weight vanishes $\omega_{vw} =0$, see \cref{lastpicture} Note that a weighted graph is fully characterized by its adjacency matrix $\Gamma_{G} \defeq (\omega_{vw})_{v,w\in V}$.

To any $d$-weighted graph $(G,\omega )$ one may associate a \textit{qudit graph state} $\ket{\psi_G} \in \mathcal{H}_d^{\otimes |V|}$ as follows
\begin{equation}
\label{graph_state_weighted}
\ket{\psi_G} = \prod_{(v,w)\in E} \Big( \textbf{CZ}_{\{vw\}}\Big)^{\omega_{vw}} \ket{+}^{\otimes V},
\end{equation}
where $\textbf{CZ}^{\{vw\}} \defeq  \omega_d^{ ij } \ket{ij}_{vw} \bra{ij}
$ is a controlled-Z operator acting on qubits $v, w$, and $\ket{+} \defeq \frac{1}{\sqrt{d}} (\ket{0}+\cdots +\ket{d-1})$, and $\omega_d$ is a root of unity of order $d$. 

\begin{figure}[ht!]
\begin{center}
\includegraphics[width = 0.55\textwidth]{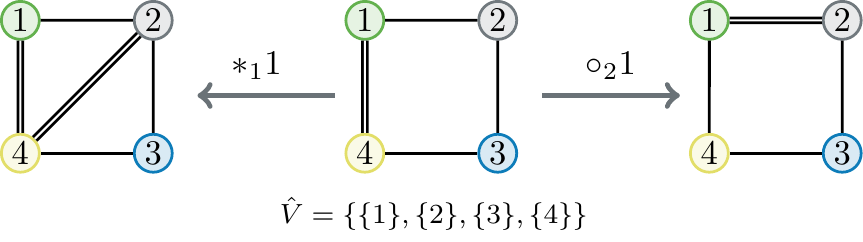}
\end{center}
\caption{\label{lastpicture}
A $3$-weighted graph on four vertices (in the middle) and LC-equivalent graphs obtained by applying two LC-operations on vertex $1$ (on the sides). The weights $0,1,2$ are graphically represented as no-edge, single-edge, and double-edge respectively. Note that the foliage partition is trivial, \ie~$|V|=|\hat{V}|$, which is impossible for graphs without weights on four vertices, see orbits 3 and 4 in \cref{tab:symmetries_table1}.}
\end{figure}

We can further introduce a generalization of local complementation for $d$-weighted graphs~\cite{Clifford}. Contrary to the qubit case, we have two distinct operators $*_a w$ and $\circ_b v$, that act on a $d$-weighted graph with a given adjacency matrix.
In the following, all operations are modulo $d$.

\begin{definition}[$*_a w$, $\circ_b v$]\label{def:qudit_local_complementation}
Consider a $d$-weighted graph with vertex set $V$ and adjacency matrix $\Gamma_{G} = (\omega_{ij})_{ij \in V}$. 
For any vertex $w$ and integer $a \in \mathbb{Z}_d$, the operator $*_a w$ transforms the graph $G$ into a new graph $G *_a w$ with the same vertex set $V$ and adjacency matrix defined by
\begin{equation}
(\Gamma_{G *_a w})_{jk} \defeq (\Gamma_{G})_{jk} + a(\Gamma_{G})_{wj}(\Gamma_{G})_{wk}
\end{equation}
for $j \neq k$, and $(\Gamma_{G}^{\prime})_{jj} = 0$ for all $j$.
Similarly, for any vertex $v$ and nonzero integer $b \in \mathbb{Z}_d$, the operator $\circ_b v$ transforms the graph $G$ into a new graph $G \circ_b v$ with an adjacency matrix obtained by pre- and post-multiplying $\Gamma_{G}$ with a diagonal matrix 
$I(v, b) \Gamma_{G} I(v, b)$,
where $I(v, b) \defeq \operatorname{diag}(1, \ldots, 1, b, 1, \ldots, 1)$ has $b$ as its $v$-th diagonal entry and 1 elsewhere.
\end{definition}

The operations $ *_a w$ and $\circ_b v$ are visualized in \cref{fig:qudit_lc,lastpicture}. Similarly to the qubit case, we have the following. 

\begin{theorem}[see \cite{vandennestGraphicalDescriptionAction2004}]
For any prime number $d$, two qudit graph states are LC-equivalent if and only if the associated $d$-weighted graphs are equivalent under a sequence of local complementations.
\end{theorem}

\begin{figure}[ht!]
\begin{center}
\includegraphics[width = 0.5\textwidth, 
    trim=16.5cm % left
        5cm % bottom
        16.5cm % right
        5cm,% top
    clip]{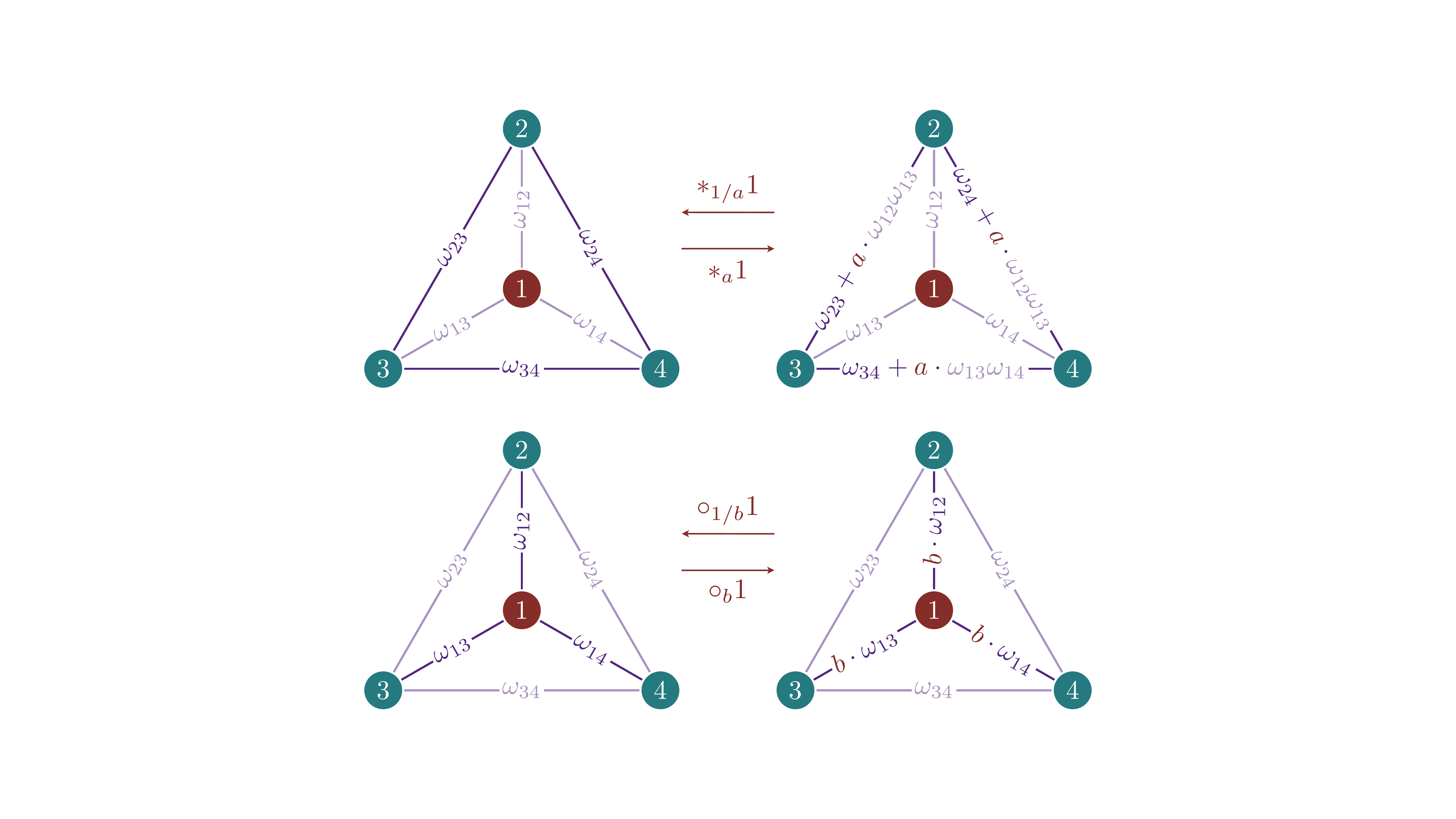}
\end{center}
\caption{\label{fig:qudit_lc}
Qudit local complementation operations $*_a 1$ and $\circ_b 1$ as introduced in \cref{def:qudit_local_complementation}.
The operation $*_a 1$ updates the weights of the edges between all pairs of vertices based on their connections to the vertex $1$ and the integer $a$.
The operation $\circ_b 1$ scales the weights of the edges between vertex $1$ and all other vertices by the integer $b$; it does not affect edges between other pairs of vertices that do not involve vertex $1$.}
\end{figure}

For $d$-weighted graphs we can define a foliage partition similarly as for graphs.

\begin{definition}
\label{def:2}
Consider a $d$-weighted graph characterized by its adjacency matrix $\Gamma_G= (\omega_{vw})_{v,w\in V}$. We define the binary relation on the set of vertices $V$. Two vertices $v,w$ are related $v \sim w$ if and only if they are in the same connected component and $\omega_{v,u_1}\cdot \omega_{w,u_2}=\omega_{v,u_2}\cdot \omega_{w,u_1}$ for any other pair of vertices $u_1,u_2 \neq v,w$.
\end{definition}

\begin{proposition}
\label{lemma2}
The relation defined in \cref{def:2} is an equivalence relation and therefore provides a set partition $\hat{V}=\{V_1,\ldots,V_k\}$ of vertices. We call it the \textit{foliage partition of} the $d$-weighted graph. 
\end{proposition}

\begin{theorem}
\label{th2}
The foliage partition of a $d$-weighted graph is invariant under any LC-transformation.
\end{theorem}

\noindent
Proofs of both statements are presented in \cref{proofs}. 
In the remaining part of this section, we investigate what are the possible structures for subsets in a foliage partition.

\begin{definition}[Strong twins]
\label{def:leaf_axil_weighted}
Two vertices $v,w \in V$ that have pairwise the same neighborhood, \ie, $N_{v} \backslash \{w\} = N_{w} \backslash \{v\}$, and weights to adjusted edges agree up to some constant, \ie~there is an $a\in \mathbb{Z}_{d}$ such that $\omega_{u v} =a \;\omega_{u w}$, for all $u\in N_{w} \backslash \{v\}$ are called strong twins.
\end{definition}

\begin{lemma}
\label{obs1}
If in part $V_i$ of the foliage partition there is at least one leaf, part $V_i$ consists of a single axil and all adjacent leaves. 
\end{lemma}

\begin{proof}
Suppose that there is a leaf $v \in V_i$. Denote by $w\in V_i$ the only vertex adjacent to $v $. Note that for any other vertices $u_1,u_2$, $\omega_{v u_1}=\omega_{v u_2}=0$, hence equation $\omega_{v u_1} \cdot \omega_{w u_2} = \omega_{v u_2} \cdot \omega_{w u_1}$
is trivially satisfied. Therefore $w \in V_i$. 
%Obviously 
Unless $|V|=2$,
$w$ is not a leaf, otherwise the graph $G$ is not connected. 

Suppose now, that there is another vertex $z\in V_i$ which is not a leaf. Note that $z$ is not connected to $v$, hence it is connected to at least one vertex $x \neq w$. Therefore $\omega_{z x} \neq 0$ and obviously $\omega_{wv} \neq 0$ as we discussed before, while $\omega_{zw} = 0$ since $v$ is a leaf. This contradicts \cref{con3} with $u_1=w$ and $u_2=x$, \ie~$\omega_{v w} \cdot \omega_{z x} = \omega_{v x} \cdot \omega_{z w}$. 

As we have shown, if there is a leaf $v \in V_i$, the only vertex adjacent to $v $ has degree at least two and also belongs to $V_i$. We shall investigate what are other vertices are elements of $V_i$. As we have shown, there are no other vertex of degree larger then one in $V_i$, hence all remaining vertices in $V_i$ are leaves. It is trivial to show that all other leaves adjacent to $w $ also belong to $V_i$. 

On the other hand, leaves adjacent to any other vertex then $w$ do not belong to $V_i$. Indeed, suppose that there are two leaves $v_1,v_2 \in V_i$ adjacent to two different vertices $w_1,w_2$. Then both of them belong to $w_1,w_2 \in V_i$, and have degrees at least two, which leads to a contradiction, and concludes the proof. 
\end{proof}

\begin{lemma}
\label{obs2}
If in part $V_i$ of the foliage partition there there are no leaves, all elements in $V_i$ have pairwise the same neighborhood, and are pairwise strong twins.
\end{lemma}

\begin{proof}
Suppose on the contrary that there are two vertices $v,w \in V_i$ with a different neighborhood. Without loss of generality, we may assume that there exists a vertex $x\neq v,w $, connected to $v$ and not connected to $w$, \ie~$\omega_{vx}\neq 0$, $ \omega_{wx}=0$. Since $w$ is not a leaf, it is connected to at least one vertex else then $v$. Therefore, there exists $z\neq v,x$ such that $\omega_{wz}\neq 0$. Since $v,w \in V_i$ are elements of the same subset in foliage partition, they satisfy $\omega_{vx}\cdot\omega_{wz} =\omega_{vz}\cdot\omega_{wx}$. Note that the left side of this equation is non-vanishing, while the rights one equals zero. This contradiction shows that both vertices $v,w$ have pairwise the same neighborhood.

Choose any vertex $u$ from the joint neighborhood of $v$ and $w$, and define $a=\tfrac{\omega_{uv}}{\omega_{uw}}$. Note that for any other vertex $u'$ from the joint neighborhood of $v$ and $w$, we have $\omega_{vu}\cdot\omega_{wu'} =\omega_{vu'}\cdot\omega_{wu}$, and hence $\omega_{vu'} =a\omega_{wu'}$, therefore $w$ and $v$ are strong twins.
\end{proof}

\begin{lemma}
\label{EitherConOrNotCon}
If in part $V_i$ of the foliage partition there there are no leaves, all elements of $V_i$ are either pairwise connected or pairwise disconnected. 
\end{lemma}

\begin{proof}
Suppose that there is a vertex $v\in V_i$ which is connected to another vertex $w \in V_i$ and disconnected from another vertex $z \in V_i$. This means that $\omega_{vw}\neq 0$, $\omega_{vz}=0$. Since there are no leaves in a graph $G$, vertex $w$ is connected to at least one vertex $x$ where $x\neq z,v$. Therefore $\omega_{wx}\neq 0$. Since $z,w \in V_i$ are elements of the same subset in foliage partition, they satisfy $\omega_{wx}\cdot\omega_{zv} =\omega_{wv}\cdot\omega_{zx}$. Note that the left side of this equation is non-vanishing, while the rights one equals zero. This contradiction shows that any vertex $v\in V_i $ is either connected to all other vertices in $V_i$ or is not connected to any of them. The statement of \cref{EitherConOrNotCon} follows easily.
\end{proof}

\begin{proposition}
\label{EquiDefi}
Each set $V_i\in \hat{V}$ in the foliage partition is of 
exactly one\footnote{In the unique case of a fully connected graph on two vertices $K_2$, the foliage partition contains one set. This set has a single element that is both of type \axilleaf~and \contwins, but this singular case should not cause any confusion later.}  of the following four types:
\begin{enumerate}
\item[\single] $V_i$ contains only one element, \ie~$|V_i|=1$, or if not:
\item[\axilleaf] $V_i$ is a star graph and it is connected to the rest of the graph only by its center,
\item[\contwins] $V_i$ is a fully connected graph and all its vertices have pairwise the same neighborhood, and are pairwise strong twins
\item[\discontwins] $V_i$ is a fully disconnected graph and all its vertices have the same  neighborhood, and are pairwise strong twins.
\end{enumerate}
Furthermore, each set $V_i\in \hat{V}$ is maximal, \ie~by adding any additional vertex to the set $V_i'\defeq V_i\cup v$ for $v\notin V_i$, $V_i'$ is not one of the above sets. 

On the other hand, any set $W\subset V$ of the above form which is maximal, belongs to the foliage partition, \ie~$W\in \hat{V}$.
\end{proposition}

\begin{proof}
In order to simplify the argument, we call the set $W\subset V$ \textit{special} if it is one of the four types listed above. Directly form \cref{obs1,obs2,EitherConOrNotCon}, each set $V_i\in \hat{V}$ in the foliage partition is a special set. On the other hand, consider any subset $W\subset V$ of vertices which is special. It is a straightforward observation that all vertices $v,w \in W$ are pairwise related $v\sim w$. 

Consider a set $V_i\in \hat{V}$ in the foliage partition, we shall see that $V_i$ is a maximal special set. For any $v\not\in V_i$, the set $V_i'=V_i\cup\{v\} \not\in \hat{V}$ is not an element of the foliage partition, and hence $W_i'$ cannot be a special set. Therefore each set $V_i \in \hat{V}$ is a maximal special set, in a sense that for any $v\not\in V_i$, the set $V_i'=V_i\cup\{v\} \not\in \hat{V}$ is not special.

Conversely, any maximal special set $W\subset V$ belongs to the foliage partition. Indeed, as we demonstrated, all vertices $v,w\in W$ are pairwise related $v\sim w$. On the other hand, for any vertex $t\not\in W$, the set $W\cup \{t\}$ is not special, and hence contains vertices which are pairwise not related. As all vertices in $W$ are pairwise related, $t\not\sim w$ for some vertex $w\in W$, hence $t$ and $w$ are not in the same equivalence class. This shows that all vertices in $W$ are in one equivalence class, and all other vertices are not in this class. Hence $W\in \hat{V}$ is an element of foliage partition.

Lastly, note that all four types of the special sets are pairwise excluding, meaning that being one type precludes the possibility of being another type, except the unique case of a fully connected graph on two vertices $K_2$, for which the foliage partition contains one set which is both of type \axilleaf~and \contwins.
\end{proof}

\noindent
Note that \cref{EquiDefi} reduces to \cref{def:foliage_partition} in the special case $d=2$.

\section{Foliage partition is well defined and invariant under LC operations}
\label{proofs}

In this section, we show that the foliage partition is well-defined and LC-invariant. Firstly, we show that relation $\sim$ defined in \cref{def:relation} for graphs, and more generally, in \cref{def:2} for weighted graphs are, indeed, equivalence relations. 

\begin{proof}[Proof of \cref{lemma2}]
Without loss of generality, we might assume that graph $G$ is connected. Indeed, for disconnected graphs, vertices from different connected component are never related, hence the relation splits into relations on the separate connected components. 

It is straightforward to see that the relation $\sim$ is reflexive and symmetric. We shall see that it is also transitive, \ie~
$
v \;\sim\; w, \; 
w \;\sim\; z \;
\Rightarrow\; 
v \;\sim\; z.
$  
Suppose that the following conditions holds
\begin{align}
\label{con1}
\forall_{x,y\neq v,w} \quad
\omega_{v x} \cdot \omega_{w y} = \omega_{v y} \cdot \omega_{w x}, \\
\label{con2}
\forall_{x,y\neq w,z} \quad
\omega_{w x} \cdot \omega_{z y} = \omega_{w y} \cdot \omega_{z x};
\end{align}
we shall show that 
\begin{equation}
\label{con3}
\forall_{u_1,u_2\neq v,z} \quad
\omega_{v u_1} \cdot \omega_{z u_2} = \omega_{v u_2} \cdot \omega_{z u_1}.
\end{equation}
Choose any $u_1,u_2\neq v,z$. We shall consider two cases, either $u_1,u_2\neq w$ or $u_1 =w$ (or by similarity $u_2 =w$). 

Case I, $u_1,u_2\neq w$. We shall show that 
\cref{con3} holds.
%$\omega_{v u_1} \cdot \omega_{z u_2} = \omega_{v u_2} \cdot \omega_{z u_1}$. 
Consider two possibilities, either both sides of this equation vanish and it is trivially satisfied, or one of the sides is non-vanishing.  Without loss of generality, suppose $\omega_{v u_1} \cdot \omega_{z u_2} \neq 0$. Therefore $\omega_{v u_1} \neq 0 \wedge \omega_{z u_2} \neq 0$. 

We shall see, that $\omega_{w u_1} \neq 0 \wedge \omega_{w u_2} \neq 0$. Indeed, suppose on the contrary that $\omega_{w u_1}=0$. By \cref{con1} with $x= u_1,y=u_2$, we have $\omega_{w u_2}=0$ since $\omega_{v u_1} \neq 0$. Further, taking \cref{con1} with $x= u_1,y=z$, we have $\omega_{w z}=0$ also since $\omega_{v u_1} \neq 0$. 
Finally, by \cref{con2} with $x= v,y=u_2$, we have $\omega_{w v}=0$ since $\omega_{z u_2} \neq 0$. Therefore vertex $w$ is not connected to any of $v,z,u_1,u_2$. However, since the graph is connected, there exists another vertex $u_3$, such that $\omega_{w u_3} \neq 0$. Observe that this contradicts \cref{con2} with $x=u_2,y=u_3$, since $\omega_{w u_3} \neq 0$ and $\omega_{z u_2} \neq 0$ on the right hand side, whereas $\omega_{w u_2} = 0$ on the left hand side. An analogous argument applies when assuming $\omega_{w u_2}=0$ instead of $\omega_{w u_1}=0$.

As we have shown, if $\omega_{v u_1} \cdot \omega_{z u_2} \neq 0$, then $\omega_{w u_1} \neq 0$ and $ \omega_{w u_2} \neq 0$. Taking \cref{con1} and \cref{con2} with $x=u_1, y=u_2$, we have $\omega_{v u_1} \cdot \omega_{w u_2} = \omega_{v u_2} \cdot \omega_{w u_1}$ and $\omega_{w u_1} \cdot \omega_{z u_2} = \omega_{w u_2} \cdot \omega_{z u_1}$, after multiplying by sides, we get 
$
\omega_{v u_1} \cdot \textcolor{violet}{\omega_{w u_2}} \cdot \textcolor{green}{\omega_{w u_1}} \cdot \omega_{z u_2} = \omega_{v u_2} \cdot \textcolor{green}{\omega_{w u_1}} \cdot \textcolor{violet}{\omega_{w u_2}} \cdot \omega_{z u_1}.
$. 
Note that \textcolor{green}{$\omega_{w u_1}$} and \textcolor{violet}{$\omega_{w u_2}$} appear on both sides of the equation above, and as we have shown both are non-zero. Therefore, we conclude that $\omega_{v u_1} \cdot \omega_{z u_2} = \omega_{v u_2} \cdot \omega_{z u_1}$, which finishes the proof in Case I.

Case II, $u_1 = w$. We shall show that \cref{con3} holds, that is, $\omega_{v w} \cdot \omega_{z u_2} = \omega_{v u_2} \cdot \omega_{z w}$.

By substituting
$x=z$, $y=u_2$ in \cref{con1}
and
$x=v$, $y=u_2$ in \cref{con2},
we have 
$\textcolor{green}{\omega_{v z}} \cdot \textcolor{violet}{\omega_{w u_2}} 
= 
\omega_{v u_2} \cdot \omega_{w z}$, 
and 
$\omega_{w v} \cdot \omega_{z u_2} 
= 
\textcolor{violet}{\omega_{w u_2}} \cdot \textcolor{green}{\omega_{z v}}$
and thus
$
\omega_{w v} \cdot \omega_{z u_2}
%=\omega_{v z} \cdot \omega_{w u_2} 
= \omega_{v u_2} \cdot \omega_{w z},
$ 
which completes the proof of Case II
since the weights $\omega_{a b} = \omega_{b a}$
are symmetric.

The equivalence classes with respect to the equivalence relation $\sim$ imply a partition for any $d$-weighted graph $(G, \omega)$.
\end{proof}

Secondly, we show that foliage partition is invariant under local complementary operations. This is the statement of \cref{th2}, which become \cref{th1q} in the very special case $d=2$. 

\begin{proof}[Proof of \cref{th2}]
Without loss of generality, we might assume that all graphs are connected. Note that property of being a connected graph is invariant under LC-transformations. Indeed, suppose that the graph has two connected components $C_1$ and $C_2$. Applying any LC-transformation on $C_1$ component ($C_2$ equivalently) does not create any connection between $C_1$ and $C_2$. Therefore LC-transformations cannot decrease the number of connected components. Since they are invertible, they cannot increase this number either.

Suppose that two vertices $v,w$ belong to the same subset in a foliage partition. It means that for any other pair of vertices $u_1,u_2\neq v,w$ the following holds $\omega_{v u_1} \cdot \omega_{wu_2} = \omega_{v u_2} \cdot \omega_{w u_1}$. We shall see that this is preserved under any LC-transformation. We shall consider two types of operations: $\circ_a$ and ${*}_a$ acting on all possible vertices (see Definition~\ref{def:qudit_local_complementation}). 
Without loss of generality, we can consider operations acting on $v, u_1$ and any other vertex $z\neq v,w,u_1,u_2$. 

For clarity of presentation, we denote by $\omega_{x,y}$ the weights of an initial graph, and by $\omega_{x,y}'$ the weights of the transformed graph. We shall show that 
\begin{equation}
\omega_{v u_1}' \cdot \omega_{w u_2}' = \omega_{v u_2}' \cdot \omega_{w u_1}'
\label{con4}
\end{equation} 
assuming 
\begin{equation}
\omega_{v x} \cdot \omega_{w y} = \omega_{v y} \cdot \omega_{w x}
\label{con5}
\end{equation}
for all $x,y\neq v,w$ in all six distinct cases.

Case I: $\circ_a v$. Note that in this case $\omega_{v u_1}'=a \omega_{v u_1}$, $\omega_{w u_2}'=\omega_{w u_2}$, $\omega_{v u_2}' =a \omega_{v u_2}$, and $ \omega_{w u_1}'=\omega_{w u_1}$. Hence \cref{con4} is trivially satisfied by \cref{con5} with $x=u_1$ and $y=u_2$ for any $a\neq 0$. 

Case II: $\circ_a u_1$ and Case III: $\circ_a z$, are analogous to the previously considered Case I.

Case IV: ${*}_a v$. Note that in this case the following weights remain unchanged: $\omega_{v u_1}'= \omega_{v u_1}$, and $\omega_{v u_2}' = \omega_{v u_2}$, while two other weights changed in the following way: 
$\omega_{w u_2}'=\omega_{w u_2} +a\omega_{v u_2}\omega_{w v}$, and $ \omega_{w u_1}'=\omega_{w u_1}+a\omega_{v u_1}\omega_{w v}$. Hence the left-hand side of \cref{con4} equals $\omega_{v u_1}\omega_{w u_2} + a\omega_{v u_2}\omega_{w v} \omega_{v u_1}$, while the right-hand side equals $\omega_{v u_2}\omega_{w u_1}+a\omega_{v u_1}\omega_{w v}\omega_{v u_2}$. Note that in both expressions there is the same term $a\omega_{v u_2}\omega_{w v} \omega_{v u_1}$, furthermore $\omega_{v u_1}\omega_{w u_2}=\omega_{v u_2}\omega_{w u_1}$ by \cref{con5} with $x=u_1, y=u_2$. This finishes the proof in this case.

Case V: ${*}_a u_1$ is analogous to the previously considered Case IV.

Case VI: ${*}_a z$. In this case weights change in the following way: 
$\omega_{v u_1}'=\omega_{v u_1} +a \omega_{v z}\omega_{z u_1} $, 
$\omega_{v u_2}'=\omega_{v u_2}+ a\omega_{v z}\omega_{z u_2}$, 
$\omega_{w u_2}'=\omega_{w u_2} +a\omega_{w z}\omega_{z u_2}$, and 
$ \omega_{w u_1}'=\omega_{w u_1}+a\omega_{w z}\omega_{z u_1}$. Therefore the left-hand side and the right-hand side of \cref{con4} are of the following form 
\begin{align*}
L\;=\;&
\omega_{v u_1}\omega_{w u_2}+
a \Big(\omega_{v z}\omega_{z u_1}\omega_{w u_2}+\omega_{w z}\omega_{z u_2}\omega_{v u_1} \Big)
\\
&+a^2 \omega_{v z}\omega_{z u_1} \omega_{w z}\omega_{z u_2}
\\
R\;=\;&
\omega_{v u_2}\omega_{w u_1}+
a \Big( \omega_{v z}\omega_{z u_2}\omega_{w u_1}+\omega_{w z}\omega_{z u_1}\omega_{v u_2}\Big)
\\
&+a^2 \omega_{v z}\omega_{z u_2}\omega_{w z}\omega_{z u_1},
\end{align*}
respectively. Note that $\omega_{v u_1}\omega_{w u_2}=\omega_{v u_2}\omega_{w u_1}$, hence the first terms in both expressions are equal. Furthermore terms with the $a^2$ prefactors are trivially equal. Finally, applying \cref{con5} with $x=z$, $y=u_2$ in the first summand, and with $x=z$, $y=u_1$ in the second summand, we have $\omega_{v z}\omega_{z u_1}\omega_{w u_2}+\omega_{w z}\omega_{z u_2}\omega_{v u_1}=\omega_{v u_2}\omega_{z u_1}\omega_{w z}+\omega_{w u_1}\omega_{z u_2}\omega_{v z}$, which shows that $L=R$ and finishes the proof. 
\end{proof}

\section{Summary}

In this paper, we introduced a simple and efficiently computable local Clifford (LC)-invariant for graph states: For any graph $G=(V,E)$ we defined the \textit{foliage partition} $\hat{V}$ of its vertices.
We demonstrated that the foliage partition $\hat{V}$ is invariant under local complementation transformations of the graph, meaning that it is invariant under LC-operations on the associated graph state. This implies that having the same foliage partition is a necessary condition for two graph states to be LC-equivalent.

To better understand the foliage partition, we analyzed it in terms of structural properties of the graph. We showed that vertices belonging to the same subset in the foliage partition exhibit similar neighborhood properties. Importantly, each set in a foliage partition is either a star connected with the rest of a graph by its center only, or a fully connected/fully disconnected set of vertices that have pairwise the same neighborhood. This allowed us to define a \textit{type function} $T$ encoding these neighborhood properties of sets in $\hat{V}$. 

We then investigated the relationship between the foliage partition of a graph and the entanglement properties of the associated graph state. Specifically, we found that the foliage partition of $G$ is trivial if and only if all $2$-body marginals of $\ket{G}$ are maximally mixed, \ie, if and only if $\ket{G}$ is $2$-uniform.

With the foliage partition $\hat{V}$ of a given graph $G$ as a foundation, we defined its \textit{foliage graph} $\hat{G}$ as a graph on the vertices $\hat{V}$. 
We proved that the adjacency matrix $\Gamma_{\hat{G}}$ of a foliage graph $\hat{G}$, together with the type function $T$, encapsulates all information about the entanglement of a graph state $\ket{G}$. For graphs with a non-trivial foliage partition, this simplifies the well-known result that the entanglement properties of $\ket{G}$ are encoded in its adjacency matrix $\Gamma_{G}$.

The foliage graph $\hat{G}$, together with the type function $T$, and the set of the graph's axils $A$, represents the underlying graph $G$ in a compact way. We termed this tuple $(\hat{G},T,A)$ the \textit{foliage representation} of $G$.

\new{Further, we show that the foliage partition provides a lower bound on the number of LC-classes $\mathcal{C}(n)$, which is exponential in $\sqrt{n}$. Notably, the lower bound on the number of such classes following from the foliage of a graph is linear in $n$. This shows the advantage of using the foliage partition over the foliage itself.}

We then discussed the \textit{LC-automorphism group} $\AutG$, the group of all permutations of vertices of a given graph $G$ leading to LC-equivalent graphs. We provided satisfactory bounds on the size of $\AutG$ group in terms of the foliage partition. 

Lastly, we extended our results to qudit graph states, defining the foliage partition for weighted graphs and proving its invariance under the generalized qudit local complementation operations. 

In conclusion, the foliage partition is a practical, easy-to-compute LC-invariant for graph states, capturing both the structure of graphs and the $2$-body marginal properties of their related graph states. It is instrumental in determining LC-symmetries of graphs, in finding the number of LC-orbits, and the sizes of LC-orbits for both labeled and unlabeled graphs. It allows us to define the foliage representation of graphs, which presents these graphs in a compact way. 

\section*{Acknowledgement}
We thank Lina Vandré, Jarn de Jong, Otfried Gühne and Anna Pappa for useful discussions. 
\new{We would like to thank Nathan Claudet and Nikolai Wyderka for an interesting email exchange, and the two anonymous reviewers whose comments helped to improve this manuscript.} 
This research was supported by an NWO Vidi grant (Project No. VI.Vidi.192.109). This work is funded/co-funded by the European Union (ERC, ASC-Q,
101040624). Views and opinions expressed are however those of the authors
only and do not necessarily reflect those of the European Union or the
European Research Council. Neither the European Union nor the granting
authority can be held responsible for them.  F.~H.~acknowledges support from the Emmy Noether DFG grant No.~418294583 and from the European Union via the Quantum Internet Alliance project. 
We thank the organizers of the \lq\lq 5th Seefeld Workshop on Quantum Information\rq\rq~where this project started.

\bibliographystyle{unsrtnat}  % Orders citations by appearance
\bibliography{Physics}

\begin{figure*}
\renewcommand*\thesubfigure{\arabic{subfigure}}
\centering
\foreach \x in {1,...,146}{
	\begin{subfigure}[b]{0.066\textwidth}
		\centering
		\includegraphics[width=\textwidth]{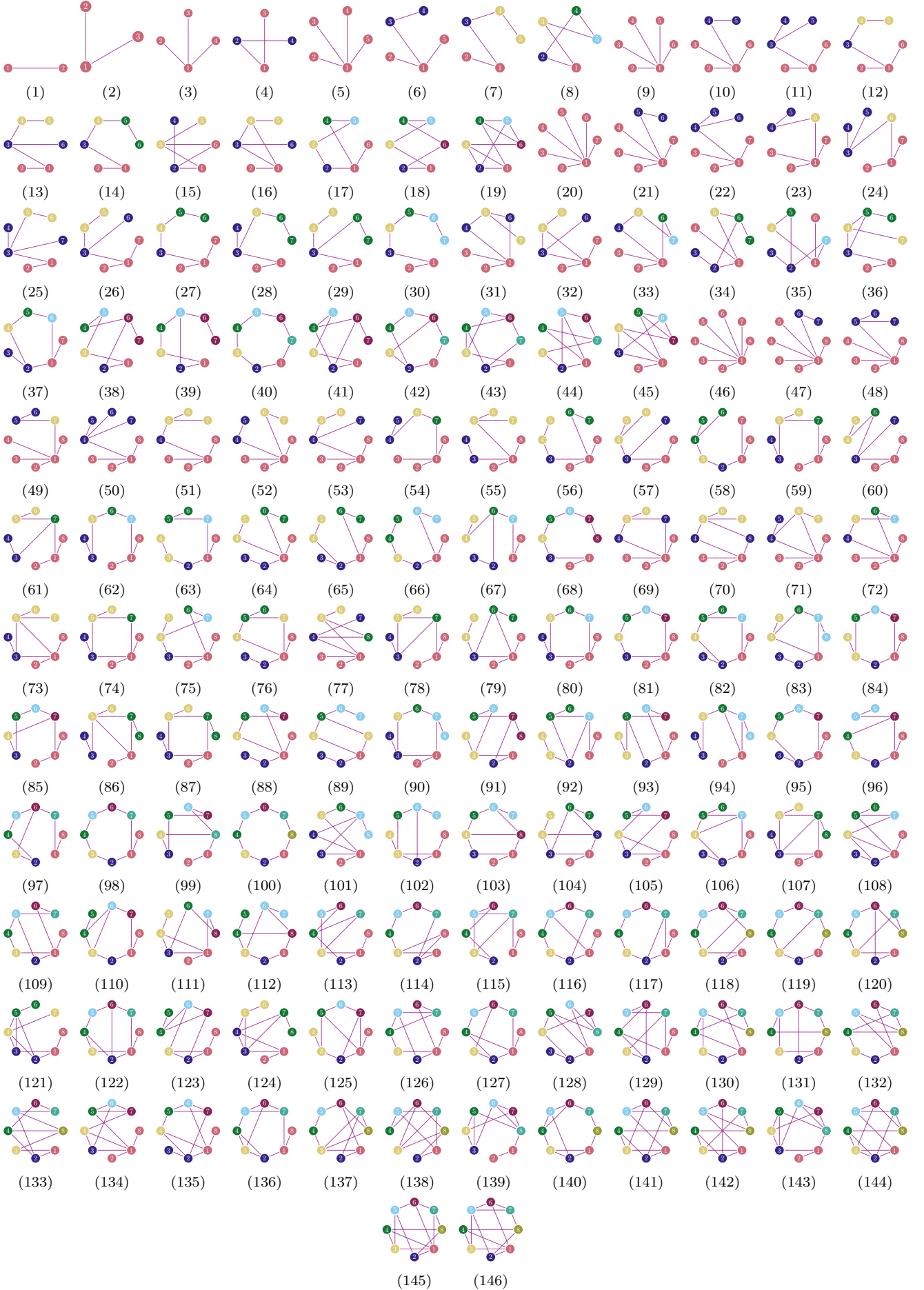}
		\caption{\label{fig:representative_\x}}
	\end{subfigure}
 }
\caption{Representatives for each local complementation orbit up to 8 qubits with highlighted foliage partition. Vertices in the same set in the foliage partition $\hat{V}$ share the same color of the node.
}
\label{fig:graph_representatives}
\end{figure*}

\begin{figure*}
\renewcommand*\thesubfigure{\arabic{subfigure}}
\centering
\foreach \x in {1,...,146}{
	\begin{subfigure}[b]{0.066\textwidth}
		\centering
		\includegraphics[width=\textwidth]{Orbit_representatives_foliage_representation_level_1/\x.pdf}
		\caption{\label{fig:representative_level1_\x}}
	\end{subfigure}
 }
\caption{
The foliage graph $\hat{G}$ for each of the graphs $G$ shown in Figure~\ref{fig:graph_representatives} -- again with highlighted foliage partition.
}
\label{fig:graph_representatives_level1}
\end{figure*}

\begin{table*}
	\includegraphics[width=0.98\textwidth]{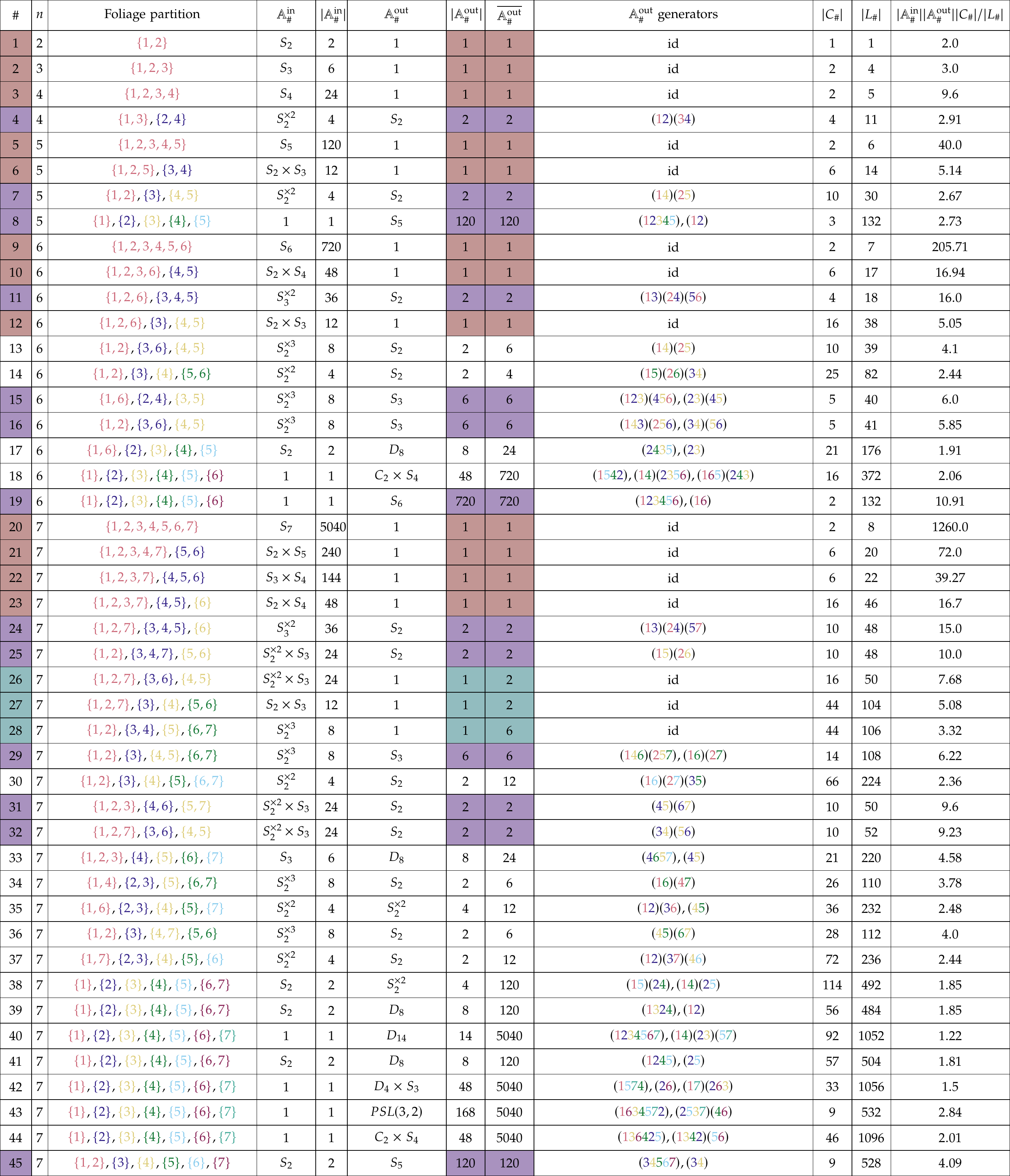}
\caption{
Symmetries for all LC-orbits up to 7 qubits.
The rows correspond to the LC-orbits of~\cref{fig:graph_representatives}. 
The first column contains the number $(\#)$ of the orbit, the second contains the number of nodes $n$ and the third contains the elements of the foliage partition $\hat{V}$.
The foliage partition $\hat{V}=\{V_1,\ldots, V_k\}$ defines the inner symmetry group $\AutGin \defeq S_{|V_{1}|}\times S_{|V_{2}|} \times \cdots \times S_{|V_{k}|}$.
All further LC-symmetries exchange different foliage partition sets of the same size, resulting in the outer symmetry group $\AutGout$. 
The foliage partition gives an upper bound $\AutGoutupp$ on the group $\AutGout$. Thus, if $|\AutGoutupp |=1$, the foliage partition completely determines $\AutGout$ group, and thereby all LC-symmetries of $G$, \textcolor{brown}{highlighted in brown}. Otherwise, the size of the group $\AutGout$ is bounded by $1\leq |\AutGout |\leq |\AutGoutupp |$. For orbits where $|\AutGout|$ matches the \textcolor{green}{lower}/\textcolor{violet}{upper} bound the row is highlighted in \textcolor{green}{green}/\textcolor{violet}{violet}. 
The remaining columns list the generators of the outer symmetry group $\AutGout$, the number of LC-equivalent graphs in each orbit counting ($|L_{\#}|$) and not counting ($|C_{\#}|$) isomorphic graphs to be different, as well as the average size of the automorphism group in the orbit --  computed as 
$|\AutGin| |\AutGout| |C_{\#}| / |L_{\#}|$.}
\label{tab:symmetries_table1}
\end{table*}

\begin{table*}
	\includegraphics[width=0.98\textwidth]{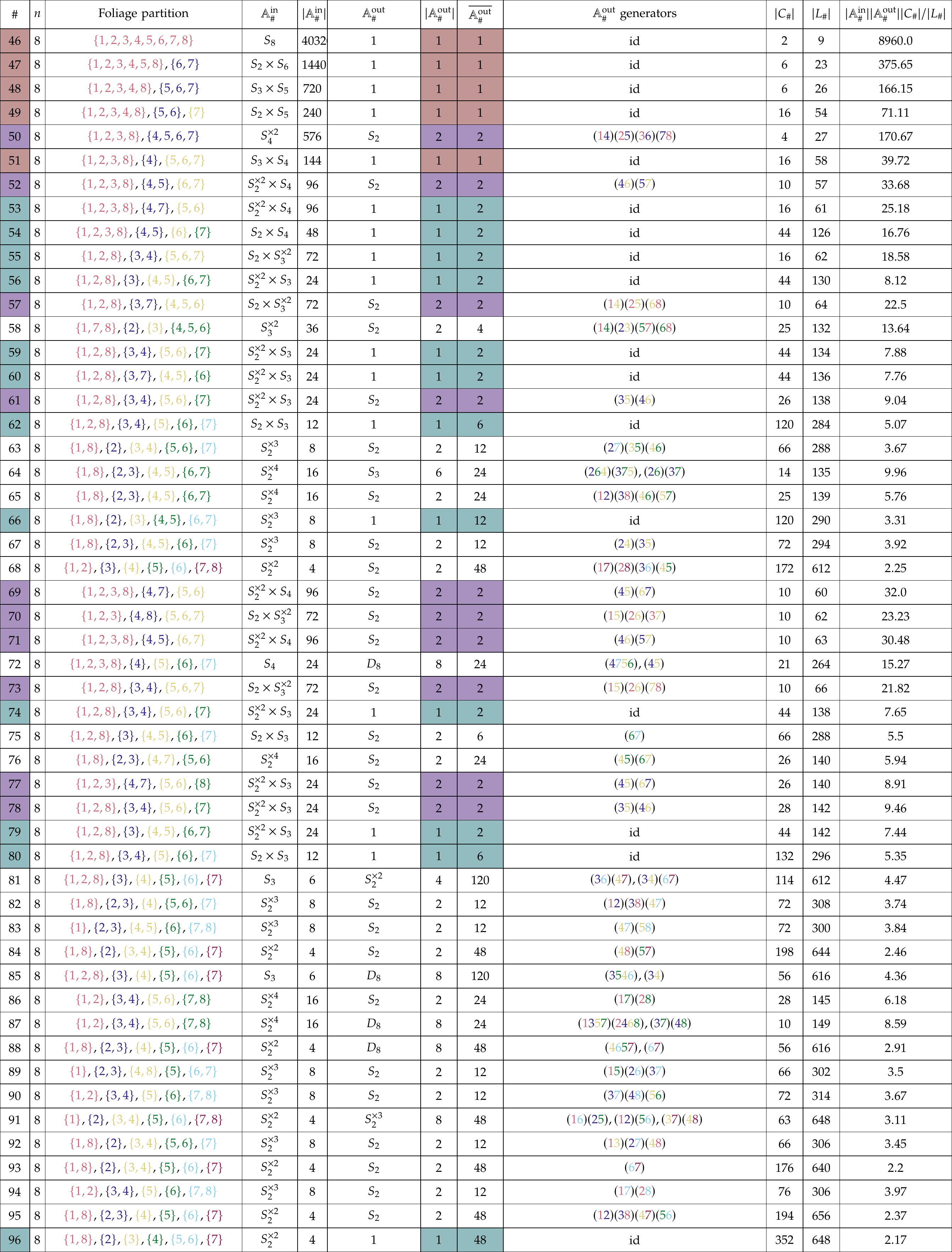}
\caption{First part of symmetries for each LC-orbit of 8 qubits. Column headings identical to~\cref{tab:symmetries_table1}.
}
\label{tab:symmetries_table2}
\end{table*}

\begin{table*}
	\includegraphics[width=0.98\textwidth]{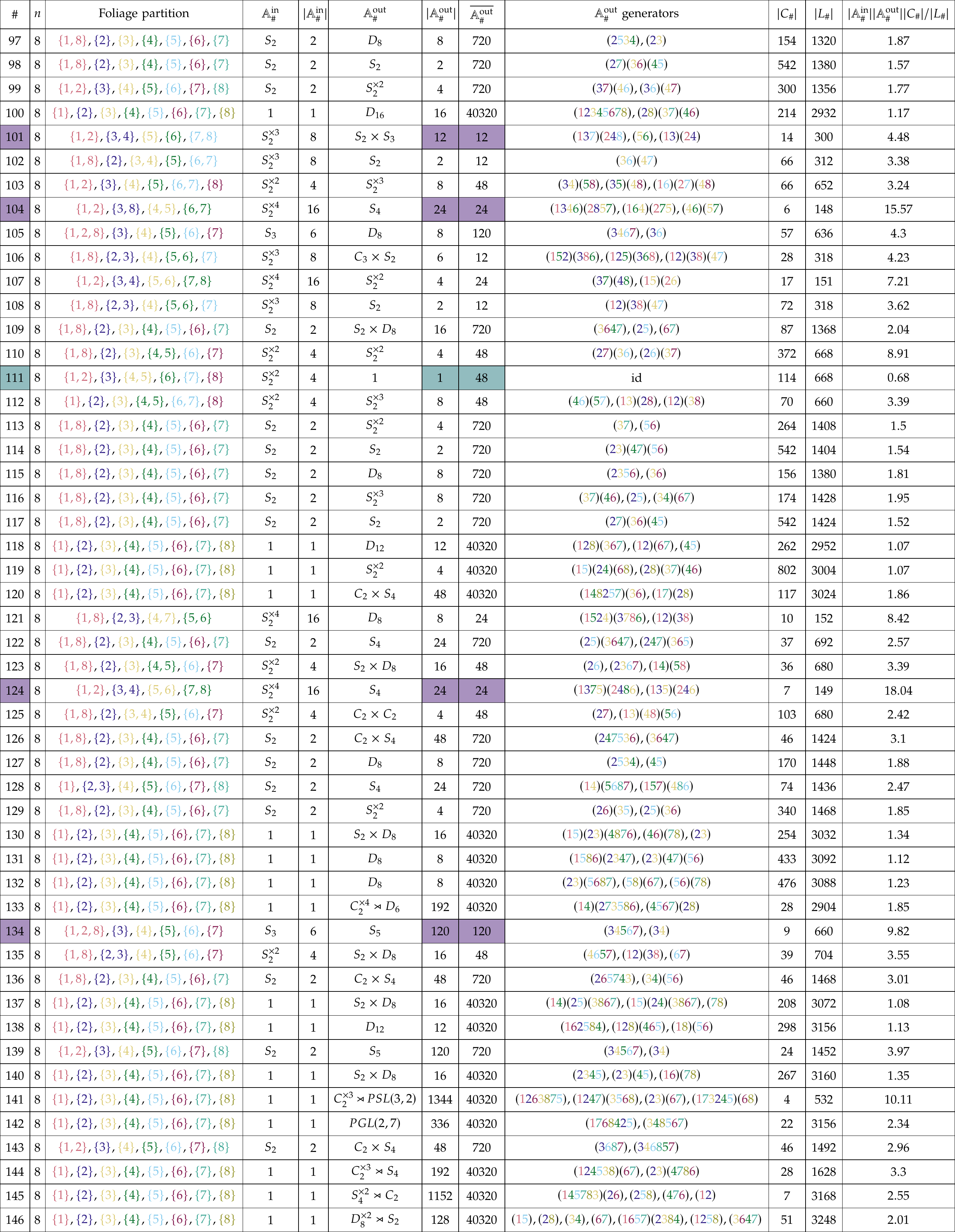}
\caption{Second part of symmetries for each LC-orbit of 8 qubits.
Column headings identical to~\cref{tab:symmetries_table1}.
}
\label{tab:symmetries_table3}
\end{table*}

\FloatBarrier

\end{document}